\documentclass[pra,aps,showpacs,twocolumn,twoside,superscriptaddress]{revtex4}
\usepackage{amsmath,amsfonts,amssymb,caption,color,epsfig,graphics,graphicx,latexsym,mathrsfs,revsymb,theorem,url,verbatim,epstopdf}

\usepackage{hyperref}

\newtheorem{definition}{Definition}
\newtheorem{proposition}[definition]{Proposition}
\newtheorem{lemma}[definition]{Lemma}

\newtheorem{theorem}[definition]{Theorem}
\newtheorem{corollary}[definition]{Corollary}
\newtheorem{conjecture}[definition]{Conjecture}

\newtheorem{remark}[definition]{Remark}
\newtheorem{example}[definition]{Example}
\newtheorem{question}[definition]{Question}

\def\squareforqed{\hbox{\rlap{$\sqcap$}$\sqcup$}}
\def\qed{\ifmmode\squareforqed\else{\unskip\nobreak\hfil
\penalty50\hskip1em\null\nobreak\hfil\squareforqed
\parfillskip=0pt\finalhyphendemerits=0\endgraf}\fi}
\def\endenv{\ifmmode\;\else{\unskip\nobreak\hfil
\penalty50\hskip1em\null\nobreak\hfil\;
\parfillskip=0pt\finalhyphendemerits=0\endgraf}\fi}
\newenvironment{proof}{\noindent \textbf{{Proof.~} }}{\qed}
\def\Dbar{\leavevmode\lower.6ex\hbox to 0pt
{\hskip-.23ex\accent"16\hss}D}
\makeatletter
\def\url@leostyle{%
  \@ifundefined{selectfont}{\def\UrlFont{\sf}}{\def\UrlFont{\small\ttfamily}}}
\makeatother
\def\bcj{\begin{conjecture}}
\def\ecj{\end{conjecture}}
\def\bcr{\begin{corollary}}
\def\ecr{\end{corollary}}
\def\bd{\begin{definition}}
\def\ed{\end{definition}}
\def\bea{\begin{eqnarray}}
\def\eea{\end{eqnarray}}
\def\bem{\begin{enumerate}}
\def\eem{\end{enumerate}}
\def\bex{\begin{example}}
\def\eex{\end{example}}
\def\bim{\begin{itemize}}
\def\eim{\end{itemize}}
\def\bl{\begin{lemma}}
\def\el{\end{lemma}}
\def\bma{\begin{bmatrix}}
\def\ema{\end{bmatrix}}
\def\bpf{\begin{proof}}
\def\epf{\end{proof}}
\def\bpp{\begin{proposition}}
\def\epp{\end{proposition}}
\def\bqu{\begin{question}}
\def\equ{\end{question}}
\def\br{\begin{remark}}
\def\er{\end{remark}}
\def\bt{\begin{theorem}}
\def\et{\end{theorem}}

\def\btb{\begin{tabular}}
\def\etb{\end{tabular}}

\newcommand{\nc}{\newcommand}

\def\a{\alpha}
\def\b{\beta}
\def\g{\gamma}
\def\d{\delta}
\def\e{\epsilon}

\def\z{\zeta}

\def\m{\mu}

\def\r{\rho}
\def\s{\sigma}

\def\ps{\psi}

 \nc{\bbA}{\mathbb{A}} \nc{\bbB}{\mathbb{B}} \nc{\bbC}{\mathbb{C}}
 \nc{\bbD}{\mathbb{D}} \nc{\bbE}{\mathbb{E}} \nc{\bbF}{\mathbb{F}}
 \nc{\bbG}{\mathbb{G}} \nc{\bbH}{\mathbb{H}} \nc{\bbI}{\mathbb{I}}
 \nc{\bbJ}{\mathbb{J}} \nc{\bbK}{\mathbb{K}} \nc{\bbL}{\mathbb{L}}
 \nc{\bbM}{\mathbb{M}} \nc{\bbN}{\mathbb{N}} \nc{\bbO}{\mathbb{O}}
 \nc{\bbP}{\mathbb{P}} \nc{\bbQ}{\mathbb{Q}} \nc{\bbR}{\mathbb{R}}
 \nc{\bbS}{\mathbb{S}} \nc{\bbT}{\mathbb{T}} \nc{\bbU}{\mathbb{U}}
 \nc{\bbV}{\mathbb{V}} \nc{\bbW}{\mathbb{W}} \nc{\bbX}{\mathbb{X}}
 \nc{\bbZ}{\mathbb{Z}}

 \nc{\bA}{{\bf A}} \nc{\bB}{{\bf B}} \nc{\bC}{{\bf C}}
 \nc{\bD}{{\bf D}} \nc{\bE}{{\bf E}} \nc{\bF}{{\bf F}}
 \nc{\bG}{{\bf G}} \nc{\bH}{{\bf H}} \nc{\bI}{{\bf I}}
 \nc{\bJ}{{\bf J}} \nc{\bK}{{\bf K}} \nc{\bL}{{\bf L}}
 \nc{\bM}{{\bf M}} \nc{\bN}{{\bf N}} \nc{\bO}{{\bf O}}
 \nc{\bP}{{\bf P}} \nc{\bQ}{{\bf Q}} \nc{\bR}{{\bf R}}
 \nc{\bS}{{\bf S}} \nc{\bT}{{\bf T}} \nc{\bU}{{\bf U}}
 \nc{\bV}{{\bf V}} \nc{\bW}{{\bf W}} \nc{\bX}{{\bf X}}
 \nc{\bZ}{{\bf Z}}

\nc{\cA}{{\cal A}} \nc{\cB}{{\cal B}} \nc{\cC}{{\cal C}}
\nc{\cD}{{\cal D}} \nc{\cE}{{\cal E}} \nc{\cF}{{\cal F}}
\nc{\cG}{{\cal G}} \nc{\cH}{{\cal H}} \nc{\cI}{{\cal I}}
\nc{\cJ}{{\cal J}} \nc{\cK}{{\cal K}} \nc{\cL}{{\cal L}}
\nc{\cM}{{\cal M}} \nc{\cN}{{\cal N}} \nc{\cO}{{\cal O}}
\nc{\cP}{{\cal P}} \nc{\cQ}{{\cal Q}} \nc{\cR}{{\cal R}}
\nc{\cS}{{\cal S}} \nc{\cT}{{\cal T}} \nc{\cU}{{\cal U}}
\nc{\cV}{{\cal V}} \nc{\cW}{{\cal W}} \nc{\cX}{{\cal X}}
\nc{\cZ}{{\cal Z}}

\nc{\hA}{{\hat{A}}} \nc{\hB}{{\hat{B}}} \nc{\hC}{{\hat{C}}}
\nc{\hD}{{\hat{D}}} \nc{\hE}{{\hat{E}}} \nc{\hF}{{\hat{F}}}
\nc{\hG}{{\hat{G}}} \nc{\hH}{{\hat{H}}} \nc{\hI}{{\hat{I}}}
\nc{\hJ}{{\hat{J}}} \nc{\hK}{{\hat{K}}} \nc{\hL}{{\hat{L}}}
\nc{\hM}{{\hat{M}}} \nc{\hN}{{\hat{N}}} \nc{\hO}{{\hat{O}}}
\nc{\hP}{{\hat{P}}} \nc{\hR}{{\hat{R}}} \nc{\hS}{{\hat{S}}}
\nc{\hT}{{\hat{T}}} \nc{\hU}{{\hat{U}}} \nc{\hV}{{\hat{V}}}
\nc{\hW}{{\hat{W}}} \nc{\hX}{{\hat{X}}} \nc{\hZ}{{\hat{Z}}}

\nc{\hn}{{\hat{n}}}

\def\dim{\mathop{\rm Dim}}

\def\lin{\mathop{\rm span}}

\def\lc{\lceil}
\def\rc{\rceil}
\def\lf{\lfloor}
\def\rf{\rfloor}

\def\ox{\otimes}

\def\sm{\setminus}

\newcommand{\ket}[1]{|#1\rangle}
\newcommand{\proj}[1]{| #1\rangle\!\langle #1 |}

\newcommand{\abs}[1]{|#1|}

\newcommand{\tbc}{\red{TO BE CONTINUED...}}

\newcommand{\red}{\textcolor{red}}
\def\Dbar{\leavevmode\lower.6ex\hbox to 0pt
{\hskip-.23ex\accent"16\hss}D}
\begin{document}
\title{$4\times4$ unextendible product basis and genuinely entangled space}

\author{Kai Wang}\email[]
{kaywong@buaa.edu.cn}
\affiliation{School of Mathematics and Systems Science, Beihang University, Beijing 100191, China}

\author{Lin Chen}\email[]{linchen@buaa.edu.cn (corresponding author)}
\affiliation{School of Mathematics and Systems Science, Beihang University, Beijing 100191, China}
\affiliation{International Research Institute for Multidisciplinary Science, Beihang University, Beijing 100191, China}

%\author{Yi Shen}\email[]
%{yishen@buaa.edu.cn}
%\affiliation{School of Mathematics and Systems Science, Beihang University, Beijing 100191, China}
%
%\author{Yize Sun}
%\affiliation{School of Mathematics and Systems Science, Beihang University, Beijing 100191, China}

\author{Lijun Zhao}
\affiliation{School of Mathematics and Systems Science, Beihang University, Beijing 100191, China}

\author{Yumin Guo}
\affiliation{School of Mathematical Sciences, Capital Normal University, Beijing 100048, China}

\date{\today}


\begin{abstract}
We show that there are six inequivalent $4\times4$ unextendible product bases (UPBs) of size eight, when we consider only  4-qubit product vectors. We apply our results to construct Positive-Partial-Transpose entangled states of rank nine. They are at the same 4-qubit, $2\times2\times4$ and $4\times4$ states, and their ranges have product vectors. One of the six UPBs turns out to be orthogonal to an almost genuinely entangled space, in the sense that the latter does not contain $4\times4$ product vector in any bipartition of 4-qubit systems. We also show that the multipartite UPB orthogonal to a genuinely entangled space exists if and only if the $n\times n\times n$ UPB orthogonal to a genuinely entangled space exists for some $n$. These results help understand an open problem in [Phys. Rev. A 98, 012313, 2018].  
\end{abstract}

\maketitle 

%\tableofcontents

\section{Introduction}
\label{sec:int}
 
The unextendible product basis (UPB) has been extensively useful in the study of positive-partial-transpose (PPT) entangled states, symmetric PPT states, Bell inequalities and fermionic system \cite{AL01,dms03,Chen2013The, Tura2012Four, Chen2014Unextendible, Augusiak2012Tight}. Recently it has been shown that there exists a non-orthogonal UPB orthogonal to a genuinely entangled (GE) subspace  \cite{PhysRevA.98.012313}. In the same paper, an open problem was proposed to ask whether the multipartite UPB orthogonal to a GE subspace exists. In this paper, we shall 
construct the $4\times4$ UPBs using the 4-qubit system. We apply the UPBs to construct PPT entangled states and 
an almost GE space, so as to approach the open problem. 

The multiqubit system can be reliably constructed in experiments \cite{Dicarlo2010Preparation,Reed2012Realization}. The multiqubit UPBs have been more and more studied theoretically   \cite{Chen2018Nonexistence,Johnston2014The,1751-8121-50-39-395301,Joh13}. Nevertheless, quantum-information tasks often deal with entangled states of high dimensions, and we need to construct UPBs of high dimensions using multiqubit UPBs. The traditional idea \cite{BDM+99} relies on the assumption that the range of constructed PPT entangled states is orthogonal to a UPB, and thus has no product state. It is an interesting problem to construct PPT entangled states using a proper subset of UPBs, so that the range of PPT entangled states has product states. Compared to the traditional idea, the construction would help create more PPT entangled states of high dimensions and more complex properties, and shows the power of UPBs we have not realized so far. This is the first motivation of this paper. 

Next, the GE state is a mixed state, which is not the convex sum of product states with respect to any bipartition of systems \cite{PhysRevLett.86.910,PhysRevLett.99.250405}. Physically, the GE state need be constructed using at least one GE pure state. The GE states such as the Greenberger-Horne-Zeilinger (GHZ) states, W states and their copies \cite{Chen2018The} play a key role in quantum  communication and computing. However it is a hard problem to determine whether a given $n$-partite state is a GE state. For $n=2$, the problem reduces to the well-known separability problem \cite{Gurvits2003Classical}. The problem has received much attentions in theory and experiment \cite{Monteiro2015Revealing, Yeo2006Teleportation, Huber2014Witnessing, Cabello2009Proposed,Kraft2018Characterizing,T2005Detecting,Huber2010Detection}. Very recently, Ref. \cite{PhysRevA.98.012313} constructed the notion of multipartite GE spaces containing only GE states. In other word, any pure state in the GE space is not a product vector with respect to any bipartition of systems. Ref. \cite{PhysRevA.98.012313} constructed a non-orthogonal UPB \footnote{The non-orthogonal UPB is a set of product vectors that are not orthogonal to any product vector at the same time} orthogonal to a GE space. It remains an open problem whether there exists a UPB orthogonal to a GE space. The positive answer of this problem would connect the two important notions, and thus motivate progress on the study of both of them theoretically and experimentally. This is the second motivation of this paper.
 
In this paper we show that there are six $4\times4$ UPBs of size eight consisting of 4-qubit product vectors. It turns out to be much harder than the construction of $4\times4$ UPBs of size $6, 7$ and $9$ consisting of 4-qubit product vectors  \cite{1810.08932}. We do not rely on the classification of 4-qubit UPBs by programming in \cite{Johnston2014The}. We apply our results to construct PPT entangled states of rank nine. They are at the same 4-qubit, $2\times2\times4$ and $4\times4$ states, and their ranges have product vectors. We further show that a family of UPB is orthogonal to an almost GE space, in the sense that the latter does not contain any $2\times2\times4$ and $4\times4$ product vector. We also show that the multipartite UPB orthogonal to a GE space exists if and only if the $n\times n\times n$ UPB orthogonal to a GE space exists for some integer $n$. These results help understand the answer to the open problem in  \cite{PhysRevA.98.012313}.

The rest of this paper is structured as follows. In Sec. \ref{sec:pre} we introduce the notions and facts such as UPBs and UOMs. For the convenience of readers, we summary our results of six $4\times4$ UPBs of size eight in Sec. \ref{sec:resultsummary}. We present two applications of our results in Sec. \ref{sec:app1} and \ref{sec:app2}, respectively. Finally we conclude in Sec. \ref{sec:con}.

\section{Preliminaries}
\label{sec:pre}

We refer to the 4-qubit subspace as $\cH_{ABCD}=\cH_A\otimes\cH_B\otimes\cH_C\otimes\cH_D=\bbC^2\otimes\bbC^2\otimes\bbC^2\otimes\bbC^2$.
For $X=A,B,C,D$, we refer to $\ket{\ps_i}\in\cH_X$ as a 2-dimensional vector.
The {\em product vector in $\cH_{ABCD}$} is a 4-partite nonzero vector of the form
$\ket{\psi_1}\ox\ket{\psi_2}\otimes\ket{\psi_3}\ox\ket{\psi_4}:=\ket{\psi_1,\ps_2,\ps_3,\psi_4}$. Suppose $\{\ket{0},\ket{1}\}$ is the computational basis in $\bbC^2$. For any alphabet say $a$, we shall refer to $\ket{a},\ket{a'}$ as a different orthonormal basis in $\bbC^2$, i.e., $\ket{a}$ is not orthogonal to $\ket{0}$ and $\ket{1}$.
One can similarly define the $n$-partite product vectors in the space $\cH=\cH_1\otimes...\otimes\cH_n$. The set of $n$-partite orthonormal product vectors $\{\ket{a_{i,1}},...,\ket{a_{i,n}}\}$ is a UPB in $\cH$ if there is no $n$-partite product vector orthogonal to the set. We shall use the following two properties of UPBs in the body and appendices of this paper. If we obtain one UPB from another by using the properties, then we say that the two UPBs are equivalent. The properties will greatly simplify the determination of UPBs. 
\begin{lemma}
\label{le:equiv}
(i) If $\{\ket{a_{i,1},...,a_{i,n}}\}_{i=1,...,m}$ is an $n$-partite UPB of size $m$ then so is $\{\ket{a_{i,\s(1)},...,a_{i,\s(n)}}\}_{i=1,...,m}$, where $\s$ is an index permutation. That is, if we switch arbitrarily the systems of a UPB then we obtain another UPB.

(ii) If $\{\ket{a_{i,1},...,a_{i,n}}\}_{i=1,...,m}$ is an $n$-partite UPB of size $m$ then so is $\{U_1\ket{a_{i,1}}\ox...\ox U_n\ket{a_{i,n}}\}_{i=1,...,m}$, where $U_1,...,U_n$ are arbitrary unitary matrices. That is, performing any product unitary transformation $U_1\ox...\ox U_n$ on a UPB produces another UPB. 
\end{lemma}

We further need the notion of \textit{unextendible orthogonal matrix (UOM)} \cite[p1]{Chen2018Multiqubit}. To understand the notion, we refer to product vectors of an $n$-qubit UPB of size $m$ as row vectors of an $m\times n$ matrix. The matrix is known as the UOM of the UPB. For orthogonal qubit states $\ket{x}$ and $\ket{x'}$ we shall refer to them as the \textit{vector variables} $x$ and $x'$ in UOMs, and vice versa. For example, the three-qubit UPB $\ket{0,0,0},\ket{1,y,z},\ket{x,1,z'},\ket{x',y',1}$ can be expressed as the UOM
\begin{eqnarray}
\label{eq:3qubit}
\bma
0&0&0\\
1&y&z\\
x&1&z'\\
x'&y'&1\\
\ema.	
\end{eqnarray}
where $x,y,z\ne0,1$. The first column of this UOM has only one \textit{independent} vector variable $x$, since $x'$ represents the qubit orthogonal to $\ket{x}$ up to global factors. Since the product vectors in the UPB are orthogonal, we say that the rows of UOM are also orthogonal. Further more we refer to the $k$'th column of a UOM as the counterpart of the $k$'th qubit of the corresponding UPB, and vice versa. So we can simply refer to the qubits of UPBs or UOMs throughout the paper. 

Furthermore, if the four-qubit UPB is still a UPB in $\bbC^2\otimes\bbC^2\otimes\bbC^4$ or $\bbC^4\otimes\bbC^4$, then we shall refer to the corresponding UOM as a UOM in $\bbC^2\otimes\bbC^2\otimes\bbC^4$ or $\bbC^4\otimes\bbC^4$, respectively. These notations will simplify our arguments in this paper.

%\subsection{linear algebra}
%\begin{lemma}
%Suppose the set $\cS$ consists of four product states $\ket{a_1,b_1},\ket{a_2,b_2},\ket{a_3,b_3},\ket{a_4,b_4}$	 in $\bbC^2\otimes\bbC^2$. Suppose the set $\cA$ of $\ket{a_1},\ket{a_2},\ket{a_3},\ket{a_4}$ has at most two identical states, and the set $\cB$ of $\ket{b_1},\ket{b_2},\ket{b_3},\ket{b_4}$ has at most two identical states. 
%
%(i) If $\cS$ spans a 3-dimensional subspace in $\bbC^2\otimes\bbC^2$ then 
%\tbc 
%\end{lemma}
%\begin{proof}
%(i) \tbc	
%\end{proof}

\section{The summary of $4\times4$ UPBs of size eight}
\label{sec:resultsummary}

We present the six matrices $F_1,F_2,...,F_6$ in \eqref{eq:UOM1-summary}-\eqref{eq:UOM7i2=i4'-summary} in Appendix \ref{app:uom,f1-f7}. The inequalities for entries in $F_i$'s are satisfied if $F_i$ corresponds to a UPB of size eight in $\bbC^4\otimes\bbC^4$. To explain the details, we construct the matrices such as $F_1(i_3=i_4')$ in \eqref{eq:f2,i3=i4'}, when the two vector variables $i_3$ and $i_4'$ are the same. We will explain with details why only $F_1,F_2,...,F_6$ may be  UOMs in $\bbC^4\otimes\bbC^4$ using 4-qubit systems in Appendix \ref{app:construct f1-f6}. 

This section consists of two parts. First we prove that $F_1,F_2,...,F_6$ are indeed UOMs in $\bbC^4\otimes\bbC^4$. Lemma \ref{le:equiv} (i) allows the operation of permuting column $1,3$ and $2,4$ at the same time, Lemma \ref{le:equiv} (ii) allows the operations (ii.a) permuting column $1,2$, and (ii.b) permuting column $3,4$. Second we show 
$F_1,F_2,...,F_6$ are inequivalent in terms of the above three operations. 

First, we assume that 
\begin{eqnarray}
\cF_j=\{\ket{a_{ji},b_{ji},c_{ji},d_{ji}},i=1,2,...,8\},	
\end{eqnarray}
is the set of product vectors defined by the matrices $F_j$, $j=1,2,...,6$. We have the following observation. 
\begin{lemma}
\label{le:fj=property}	
(i) For any $j$, any five two-qubit product vectors of the set $\{\ket{a_{ji},b_{ji}},i=1,2,...,8\}$ span $\bbC^4$; any four of the set span a subspace of dimension three or four.

(ii) For any $j$, any five two-qubit product vectors of the set $\{\ket{c_{ji},d_{ji}},i=1,2,...,8\}$ span $\bbC^4$; any four of the set span a subspace of dimension three or four.

(iii) Suppose the set of four distinct two-qubit product vectors  $\ket{a_{ji},b_{ji}},i=1,2,3,4$ spans a 3-dimensional subspace in $\bbC^4$. Then $a_{j\s(1)}=a_{j\s(2)}$ and 
$b_{j\s(3)}=b_{j\s(4)}$ for an index permutation $\s$.

(iv) Suppose the set of four distinct two-qubit product vectors $\ket{c_{ji},d_{ji}},i=1,2,3,4$ spans a 3-dimensional subspace in $\bbC^4$. Then $c_{j\s(1)}=c_{j\s(2)}$ and 
$d_{j\s(3)}=d_{j\s(4)}$ for an index permutation $\s$. 
\end{lemma}
\begin{proof}
(i) The first claim of assertion (i) can be proven by checking the first two columns of matrices $F_1,F_2,...,F_6$. In fact, there exist four linearly independent two-qubit product vectors in any five. The second claim of assertion (i) is a corollary of the first claim.

(ii) Using the similar argument to the first two columns, the first claim of assertion (ii) can be proven by checking the last two columns of matrices $F_1,F_2,...,F_6$. The second claim of assertion (ii) is a corollary of the first claim.

(iii), (iv) The assertions can be verified directly or by programming.   
\end{proof}

We are now in a position to prove that $F_1,F_2,...,F_6$ in \eqref{eq:UOM1-summary}-\eqref{eq:UOM7i2=i4'-summary} are UOMs. Suppose there exists a product vector $\ket{x,y}\in\bbC^4\otimes\bbC^4$ orthogonal to $\cF_1$. Lemma \ref{le:fj=property}	 (i) and (ii) show that $\ket{x}$ is orthogonal to four states of $\ket{a_{1i},b_{1i}},i=1,2,...,8$, and $\ket{y}$ is orthogonal to four states of $\ket{c_{1i},d_{1i}},i=1,2,...,8$. Using $F_1$,
Lemma \ref{le:fj=property} (iii) and (iv), one can show that such two sets of four states do not exist. So $\ket{x,y}$ does not exist, and $F_1$ is a UOM. One can similarly prove that $F_2,F_3,...,F_6$ are UOMs.

In the second part of this section, we explain the inequivalence of $F_1,F_2,...,F_6$. We refer readers to Table \ref{tab:f1-7} for the maximum number of independent vector variables in each columns of the UOMs. Since only the UOMs with identical number may be equivalent, we obtain that $F_4$ is not equivalent to any other $F_j$'s. Further, $F_2,F_6$ are not equivalent to any one of $F_1,F_3$ and $F_5$. Next, $F_3$ and $F_6$ are not equivalent because column 1 and 3 of $F_2$ have the submatrix $\bma0&0\\0&0\ema$, and column 1 (or 2) and 3 of $F_6$ do not.
One can similarly show that $F_1,F_3,F_5$ are not equivalent each other.

\begin{table}[h] 
$$
\begin{array}{cclc}
\text{UOM} & ~ & \text{maximum number of independent variables} & \\
\hline
F_1 && 
\quad\quad\quad\quad\quad\quad\quad 
[2,2,2,3]  \\
F_2 && 
\quad\quad\quad\quad\quad\quad\quad 
[2,2,2,4]  \\
F_3 && 
\quad\quad\quad\quad\quad\quad\quad 
[2,2,3,2]  \\
F_4 && 
\quad\quad\quad\quad\quad\quad\quad 
[2,3,2,3]  \\
F_5 && 
\quad\quad\quad\quad\quad\quad\quad 
[3,2,2,2]  \\
F_6 && 
\quad\quad\quad\quad\quad\quad\quad 
[2,2,2,4]  \\
\hline \\
\end{array}
$$
\caption{For any $x$, we say that the pair $x,x'$ contributes only one independent vector variable. For the UOM $F_1$, the array $[2,2,2,3]$ means that each of column $1,2,3$ of $F_1$ has exactly $2$ independent vector variables, and column $4$ of $F_1$ has exactly $3$ independent vector variables. The arrays for other $F_j$'s are similarly defined.}
\label{tab:f1-7}
\end{table}

\section{Application 1: constructing PPT entangled states using a proper subset of UPB}
\label{sec:app1}

In this section we construct PPT 
entangled states using the six UPBs in 
Sec. \ref{sec:resultsummary}. The traditional idea is that $\r={1\over \abs{\cS}}(I-\sum_{j\in\cS}\proj{x_j})$ is a PPT entangled state if the set of orthogonal product vectors $\{\ket{x_j},j\in\cS\}$ is a UPB, then the range of $\r$ has no product vectors. Different from the idea, we show that every UPB of the UOMs $F_1,...,F_6$ has a proper subset of cardinality seven, such that they span a subspace whose orthogonal complement space is the range of a PPT entangled state of rank nine. It sheds novel light on the construction of PPT entangled states using UPBs. This is a corollary of the following observation.
\begin{lemma}
\label{le:pptes=fewprod}
Suppose $d=d_1d_2...d_n$, and $\ket{x_1},...,\ket{x_m}\in\bbC^d=\bbC^{d_1}\otimes...\otimes\bbC^{d_n}$'s are $m$ orthonormal product vectors. 

(i) If the range of $\r={1\over d-m}(I_d-\sum^m_{j=1}\proj{x_j})$ contains at most $d-m-1$ linearly independent product vectors, then $\r$ is a PPT entangled state of rank $d-m$.

(ii) Suppose $\ket{y_1},...,\ket{y_m}\in\bbC^{d}$ are $m$ orthonormal product vectors. If the UOMs of the two sets $\ket{x_i}$'s and $\ket{y_i}$'s are locally equivalent, then the numbers of product vectors orthogonal to the two sets are the same, or are both infinite.
\end{lemma}
\begin{proof}
(i) We prove the assertion by contradiction. Suppose $\r$ is a separable state. Let $\r=\sum_ip_i\proj{a_i}$ where $\ket{a_i}\in\bbC^d$ are product vectors. So the set $\{\ket{a_i}\}$ has exactly $d-m$ linearly independent vectors. Since $\{\ket{a_i}\}$ spans the range of $\r$, the latter also has exactly $d-m$ linearly independent vectors. It is a contradiction with the hypothesis that $\cR(\r)$ has at most $d-m-1$ linearly independent vectors. So assertion (i) holds.

(ii) Let $\ket{b}$ be a product vector orthogonal to $\ket{x_i}$'s. Since the latter is locally equivalent to 	$\ket{y_i}$'s, there is a local unitary matrix $U$ such that $U\ket{x_i}=P\ket{y_{\s(i)}}$, $\forall i$ up to a vector permutation matrix $P$ and an index permutation $\s$. So $P^\dag U\ket{b}$ is a product vector orthogonal to $\ket{y_i}$'s. 

Let $X,Y$ be the numbers of product vectors orthogonal to the two sets $\{\ket{x_i}\}$ and $\{\ket{y_i}\}$, respectively. The last paragraph shows that $X\le Y$. If we switch $x_i$ and $y_i$ in the last paragraph, then the argument still holds. We have $X\ge Y$. So we have $X=Y$, and the assertion holds.
\end{proof}

In the following we apply the above lemma to constructing PPT entangled states.
By deleting the $i$'th product vector $\ket{a_{ji},b_{ji},c_{ji},d_{ji}}\in\cF_j$, we refer to 
$\cS_{ji}$ as the set of remaining seven product vectors for $i=1,2,...,8$ and $j=1,2,..,6$. That is
\begin{eqnarray}
\label{eq:csji}
\cS_{ji}=\cF_j \sm 
\{
\ket{a_{ji},b_{ji},c_{ji},d_{ji}}
\}.
\end{eqnarray}
Let $S_{ji}$ be the UOM of $\cS_{ji}$, $\cT_{ji}$ the set of $4\times4$ product vectors orthogonal to $\cS_{ji}$, and $T_{ji}$ the UOM of $\cT_{ji}$. We present the following observation.  

\begin{lemma}
\label{le:t11=3}
(i) $\abs{\cT_{11}}=4$ or $6$. The latter holds if and only if $i_3=i_4'$. 

(ii) $\abs{\cT_{21}(i_2=i_3, i_4=0)}=6$.

(iii) $\abs{\cT_{31}}=4$ or $5$. The latter holds if and only if $i_3=i_4'$. 

(iv) $\abs{\cT_{41}}=4$.

(v) $\abs{\cT_{51}}=4$ or $6$. The latter holds if and only if $i_5=i_6'$. 

(vi) $\abs{\cT_{61}}=6$.
\end{lemma} 
We refer readers to its proof in Appendix \ref{app:proof}. Now we present the main theorem of this section. The first part of following theorem from the fact that $F_1,..,F_6$ are UOMs in $\bbC^4\otimes\bbC^4$. The second part follows from Lemma \ref{le:pptes=fewprod} and \ref{le:t11=3}. 
\begin{theorem}
\label{thm:pptes=rank9}
(i) Suppose $\cF$ is one of the six sets $\cF_1,\cF_2,\cF_3,\cF_4,\cF_5,\cF_6$. The state
\begin{eqnarray}
\a={1\over8}
(I-\sum_{\ket{\ps_j}\in \cF}\proj{\ps_j})
\end{eqnarray}
is at the same time a 4-qubit, $2\times2\times4$, and $4\times4$ PPT entangled state of rank eight.

(ii) Suppose $\cS$ is one of the six sets $\cS_{11},\cS_{21}(i_2=i_3, i_4=0),\cS_{31},\cS_{41},\cS_{51}$ and $\cS_{61}(i_2=i_3)$. The state
\begin{eqnarray}
\b={1\over9}
(I-\sum_{\ket{\ps_j}\in \cS}\proj{\ps_j})
\end{eqnarray}
is at the same time a 4-qubit, $2\times2\times4$, and $4\times4$ PPT entangled state of rank nine.
\end{theorem}
We have demonstrated the idea of constructing a PPT entangled state using a proper subset of a UPB. We have used the subset in \eqref{eq:csji} by removing the first vector in $\cF_i$. One may construct more PPT entangled states in the same way, by removing one of the other vectors in $\cF_i$ for some $i$. We have found that some $\cS_{ji}$ has more than nine product vectors and thus Lemma \ref{le:pptes=fewprod} does not work here. 

Next, the states $\a$ and $\b$ in Theorem \ref{thm:geupb} are both of robust entanglement in the sense that they are entangled w.r.t. different partitions of systems. It is not close to the genuine entanglement, as the state may become unentangled if we switch the systems. This is a problem we will tackle in the next section. On the other hand from Lemma \ref{le:t11=3}, the range of states $\b$ constructed by $\cF_i$'s has product vectors. We are not sure whether the smaller subset of $\cF_i$ would generate PPT entangled states too. 

\section{Application 2: the UPB orthogonal to an almost GE space}
\label{sec:app2}

The multipartite GE 
space contains no 
bipartite product vectors w.r.t. any 
bipartition of systems 
\cite{PhysRevA.98.012313}. The GE space exists. For example, the one-dimensional subspace spanned by the multiqubit GHZ state is a GE space. Further, the GE space remains a GE space if it is multiplied by a local unitary transformation \cite{Baerdemacker2017The}. Just like the determination of GE states, characterizing the GE space of arbitrary dimension turns out to be a hard problem. 

The paper \cite{PhysRevA.98.012313}
has constructed nonorthogonal UPBs orthogonal to a GE space. It remains an open problem whether the UPB orthogonal to a GE space exists.     
In this section we construct a $4$-qubit UPB whose orthogonal space $\cG$ has no $4\times4$ product vectors w.r.t. any $4\times4$ bipartition of the 4-qubit space. It is known that the $2\times N$ UPB does not exist \cite{BDM+99}. So $\cG$ contains $2\times N$ product vectors, and it is not a GE space. Nevertheless, $\cG$ is still close to a GE space, and we refer to $\cG$ as an \textit{almost GE space}. In the following, the first main result of this section shows that 
the space orthogonal to the UPB $\cF_6$ in \eqref{eq:UOM7-summary} is an almost GE space. 
\begin{lemma}
\label{le:F6}
$\cF_6$ is a bipartite UPB across any one of the three partitions of systems, namely $AB:CD$, $AC:BD$ and $AD:BC$.
\end{lemma}
\begin{proof}
The bipartite partition of four-qubit system has only two cases, namely $2\times 8$ and $4\times 4$. Here we investigate only system $4\times 4$ since there does not exist bipartite UPB of $2\times n$ systems. Further, the $4\times 4$ partition for $\cF_6$ occurs in the system $AB:CD$, $AC:BD$ and $AD:BC$. We have proved $\cF_6$ in $\cH_{AB}\ox\cH_{CD}$ is a UOM in section \ref{sec:resultsummary}. The remaining task is to prove that $\cF_6$ is a UPB in $\cH_{AC}\ox\cH_{BD}$ and $\cH_{AD}\ox\cH_{BC}$. We show by contradiction that $\cF_6$ is a UPB in $\cH_{AC}\ox\cH_{BD}$, and one can similarly prove the other case.

Assume that $\ket{v}=\ket{\a,\b}\in\cH_{AC}\ox\cH_{BD}$ is orthogonal to the states $\ket{v_1},...,\ket{v_7}$, and $\ket{v_8}$ in $\cF_6$, see Fig. \ref{fig:digit1}. First, there must be three row vectors of $\cF_6$ orthogonal to $\ket{v}$ on system $AC$. Then we denote them by $\ket{v_1}, \ket{v_2}, \ket{v_3}$. Similarly, there must be three row vectors $\ket{v_4}, \ket{v_5}, \ket{v_6}$ orthogonal to $\ket{v}$ on system BD. Next, the space spanned by any four of $\ket{v_1}, \ket{v_2}, \cdots, \ket{v_8}$ of system $AC$ or $BD$ has dimension three or four from Lemma \ref{le:fj=property}. That is to say, there is at most one of the remaining two vectors of $\cF_6$ orthogonal to $\ket{v}$. If there is one among the two row vectors orthogonal to $\ket{v}$ then we denote it by $\ket{v_7}$. Further, $\ket{v_8}$ is not orthogonal to $\ket{v}$. It is a contradiction with the assumption the beginning of this paragraph. 
\end{proof}

\begin{figure}[htbp]
\centering
\includegraphics[width=0.48\textwidth]{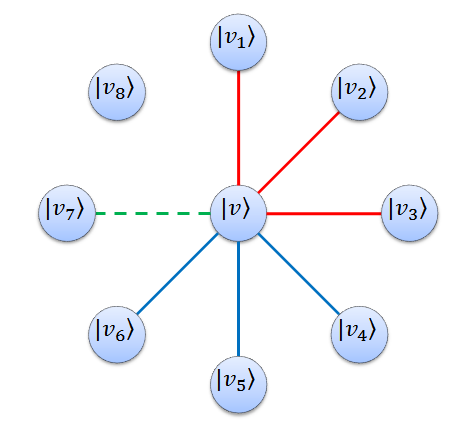}
\caption{The red line between the two product states $\ket{v}$ and $\ket{v_i}$ means that they are orthogonal on system $AC$. The blue line implies that the product vectors are orthogonal on system $BD$. The dashed line means that $\ket{v}$ may be orthogonal to $\ket{v_7}$. There is no line between $\ket{v}$ and $\ket{v_8}$, and it means that $\ket{v}$ and $\ket{v_8}$ are not orthogonal.}
\label{fig:digit1}
\end{figure}

In the remaining of this section, we investigate the properties of UPBs orthogonal to a GE space, if such UPBs exist. For simplicity we  
shall refer to such UPBs as \textit{GEUPB}. We present the following observation. 
\begin{lemma}
\label{le:geupb}
(i) The multipartite UPB is a GEUPB if and only if it is a bipartite UPB across any patition of systems.

(ii) The tensor product of two $m$-partite UPBs is still an $m$-partite UPB, and vice versa. 

(iii) The tensor product of two $m$-partite GEUPBs is still an $m$-partite GEUPB, and vice versa. 

That is, suppose
\begin{eqnarray}
&&
S=\{\ket{a_{i1},...,a_{im}}\},
\\\notag\\&&
T=\{\ket{b_{j1},...,b_{jm}}\},	
\end{eqnarray}
are two $m$-partite GEUPBs in the space $\otimes^m_{i=1}\cH_i$ and $\otimes^m_{j=1}\cK_j$, respectively. Then the set
\begin{eqnarray}
\label{eq:aibi}
\{
\ket{a_{i1},b_{j1}}
\otimes
...
\otimes
\ket{a_{im},b_{jm}}
\}	
\end{eqnarray}
is an $m$-partite GEUPB of $\abs{S}\cdot \abs{T}$ product vectors in the space $\otimes^m_{i=1}(\cH_i\otimes\cK_i)$. 	
\end{lemma}
\begin{proof}
(i) The assertion follows from \cite[Remark 5]{PhysRevA.98.012313}.

(ii) The first part of assertion (ii) follows from \cite[Theorem 8]{dms03}. The "vice versa" part follows from the definition of UPBs.

(iii) The "vice versa" part follows from the definition of GEUPBs. We prove the first part of assertion (iii) as follows.

By the definition of GEUPBs, $S$ and $T$ are bipartite UPBs across any cut of systems, say $U:\bar{U}$. Assertion (ii) implies that the set in \eqref{eq:aibi} is  a bipartite UPB in the system $(U_S,U_T):(\bar{U}_S,\bar{U}_T)$. So assertion (iii) follows from assertion (i).
\end{proof}

Using the above lemmas, we present the second main result of this section. 
\begin{theorem}
\label{thm:geupb}
The multipartite GEUPB exists if and only if the $n\times n\times n$ GEUPB exists for some integer $n$. 
\end{theorem}
\begin{proof}
(i) It suffices to prove the "if" part. It follows from the definition of GE spaces that the tripartite GEUPB $S_{ABC}\subseteq\cH_A\otimes\cH_B\otimes\cH_C$ exists. So $S_{BCA}$ and $S_{CAB}$ are both tripartite GEUPBs. It follows from Lemma \ref{le:geupb} (iii) that the tensor product of $S_{ABC}$, $S_{BCA}$ and $S_{CAB}$ is also a tripartite GEUPB. Since its local systems have equal dimensions, the assertion holds. 
\end{proof}

To further investigate the problem in \cite{PhysRevA.98.012313}, the theorem shows that it suffices to find the $n\times n\times n$ GEUPB or prove its nonexistence. It is known that the 3-qubit UPB has cardinality four \cite{Bravyi2004Unextendible}, just like those in \eqref{eq:3qubit}. So the 3-qubit UPB is orthogonal to a bipartite product vector in the system $A:BC$. Thus the 3-qubit UPB is not a GEUPB, and we have $n>2$ in terms of Theorem \ref{thm:geupb}. In particular a three-qutrit UPB of size seven has been constructed \cite{dms03}, though it is evidently not a GEUPB. We need construct three-qutrit UPBs of larger size, say close to the upper bound $23$, because any multipartite PPT states of rank at most three are separable states \cite{Lin2013Separability}. 
 
\section{Conclusions}
\label{sec:con}

We have shown that there are six inequivalent $4\times4$ UPB of size eight, when we only consider 4-qubit product vectors. We have constructed entangled states that are at the same 4-qubit, $2\times2\times4$ and $4\times4$ entangled states of rank nine. The results have been obtained using the UOM. We further have shown that the UPB $\cF_6$ is orthogonal to a space containing no $4\times4$ product vector w.r.t any bipartition of 4-qubit system. We also have shown that the multipartite UPB orthogonal to a GE space exists if and only if the $n\times n\times n$ UPB orthogonal to a GE space for some integer $n$. In spite of these results, we are still unable to show whether the UPB orthogonal to a GE space exists. Another problem is to characterize PPT entangled states of rank eight or nine that are not generated by orthogonal 4-qubit product vectors.

\section*{Acknowledgments}

This work was supported by the NNSF of China (Grant No. 11871089), and the Fundamental Research Funds for the Central Universities (Grant Nos. KG12040501, ZG216S1810 and ZG226S18C1).  

\appendix

\section{The description of six four-qubit UOMs $F_1,F_2,...,F_6$}
\label{app:uom,f1-f6}

\begin{widetext}
\begin{eqnarray}
\label{eq:UOM1-summary}
&&
F_1=
\bma
0&0&0&0\\
0&1&0&1\\
1 & g_3\ne0,1 & h_3\ne0,1 &i_3\ne0,1,i_4\\
1 & g_3' & h_3 &i_4\ne0,1\\
f_5\ne0,1 & g_3' & 1 &i_4'\\
f_5 & g_3 & 1 &i_3'\\
f_5' & 0 & h_3' &1\\
f_5' & 1 & h_3' &0\\
\ema,
\\&&
\label{eq:f2,i3=i4'}
F_1(i_3=i_4')=
\bma
0&0&0&0\\
0&1&0&1\\
1 & g_3\ne0,1 & h_3\ne0,1 &i_3\ne0,1\\
1 & g_3' & h_3 &i_3'\\
f_5\ne0,1 & g_3' & 1 &i_3\\
f_5 & g_3 & 1 &i_3'\\
f_5' & 0 & h_3' &1\\
f_5' & 1 & h_3' &0\\
\ema,
\\&&
F_2=
\bma
0    & 0    & 0    & 0\\
0    & 1    & 0    & i_2\neq 0,1 \\
1    & g_3\ne0,1  & h_3\ne0,1  & i_3\ne i_4,i_4'\\
1    & g_3' & h_3  & i_4\\
f_5\ne0,1  & g_3' & 1    & i_4'\\
f_5  & g_3  & 1    & i_3'\\
f_5' & 1    & h_3' & i_2'\\
f_5' & 0    & h_3' & 1\\
\ema,
\\&&
(i) F_2(i_2=i_3, i_4\neq 0,1)=
\bma
0    & 0    & 0    & 0\\
0    & 1    & 0    & i_2\ne0,1,i_4,i_4'\\
1    & g_3\ne0,1  & h_3\ne0,1  & i_2\\
1    & g_3' & h_3  & i_4\\
f_5\ne0,1  & g_3' & 1    & i_4'\\
f_5  & g_3  & 1    & i_2'\\
f_5' & 1    & h_3' & i_2'\\
f_5' & 0    & h_3' & 1\\
\ema,
\\&&
F_2(i_2=i_3, i_4=0)=
\bma
0    & 0    & 0    & 0\\
0    & 1    & 0    & i_2\ne0,1\\
1    & g_3\ne0,1  & h_3\ne0,1  & i_2\\
1    & g_3' & h_3  & 0\\
f_5\ne0,1  & g_3' & 1    & 1\\
f_5  & g_3  & 1    & i_2'\\
f_5' & 1    & h_3' & i_2'\\
f_5' & 0    & h_3' & 1\\
\ema,
\\&&
F_2(i_2=i_3, i_4=1)=
\bma
0    & 0    & 0    & 0\\
0    & 1    & 0    & i_2\ne0,1\\
1    & g_3\ne0,1  & h_3\ne0,1  & i_2\\
1    & g_3' & h_3  & 1\\
f_5\ne0,1  & g_3' & 1    &0\\
f_5  & g_3  & 1    & i_2'\\
f_5' & 1    & h_3' & i_2'\\
f_5' & 0    & h_3' & 1\\
\ema,
\end{eqnarray}
\begin{eqnarray}
&&
(ii) F_2(i_2=i_3', i_4\neq 0,1)=
\bma
0    & 0    & 0    & 0\\
0    & 1    & 0    & i_2\neq 0,1,i_4,i_4' \\
1    & g_3\ne0,1  & h_3\neq0,1  & i_2'\\
1    & g_3' & h_3  & i_4\\
f_5\ne0,1  & g_3' & 1    & i_4'\\
f_5  & g_3  & 1    & i_2\\
f_5' & 1    & h_3' & i_2'\\
f_5' & 0    & h_3' & 1\\
\ema,
\\&&
F_2(i_2=i_3', i_4=0)=
\bma
0    & 0    & 0    & 0\\
0    & 1    & 0    & i_2\neq 0,1\\
1    & g_3\ne0,1  & h_3\neq0,1  & i_2'\\
1    & g_3' & h_3  & 0\\
f_5\ne0,1  & g_3' & 1    & 1\\
f_5  & g_3  & 1    & i_2\\
f_5' & 1    & h_3' & i_2'\\
f_5' & 0    & h_3' & 1\\
\ema,
\\&&
F_2(i_2=i_3', i_4=1)=
\bma
0    & 0    & 0    & 0\\
0    & 1    & 0    & i_2\neq 0,1 \\
1    & g_3\ne0,1  & h_3\neq0,1  & i_2'\\
1    & g_3' & h_3  & 1\\
f_5\ne0,1  & g_3' & 1    & 0\\
f_5  & g_3  & 1    & i_2\\
f_5' & 1    & h_3' & i_2'\\
f_5' & 0    & h_3' & 1\\
\ema,
\\&&
(iii) F_2(i_2=i_4, i_3\neq 0,1)=
\bma
0    & 0    & 0    & 0\\
0    & 1    & 0    & i_2\neq 0,1,i_3,i_3' \\
1    & g_3\ne0,1  & h_3\ne0,1  & i_3\\
1    & g_3' & h_3  & i_2\\
f_5\ne0,1  & g_3' & 1    & i_2'\\
f_5  & g_3  & 1    & i_3'\\
f_5' & 1    & h_3' & i_2'\\
f_5' & 0    & h_3' & 1\\
\ema,
\\&&
F_2(i_2=i_4, i_3=0)=
\bma
0    & 0    & 0    & 0\\
0    & 1    & 0    & i_2\neq 0,1\\
1    & g_3\ne0,1  & h_3\ne0,1  & 0\\
1    & g_3' & h_3  & i_2\\
f_5\ne0,1  & g_3' & 1    & i_2'\\
f_5  & g_3  & 1    & 1\\
f_5' & 1    & h_3' & i_2'\\
f_5' & 0    & h_3' & 1\\
\ema,
\\&&
F_2(i_2=i_4, i_3=1)=
\bma
0    & 0    & 0    & 0\\
0    & 1    & 0    & i_2\neq 0,1\\
1    & g_3\ne0,1  & h_3\ne0,1  & 1\\
1    & g_3' & h_3  & i_2\\
f_5\ne0,1  & g_3' & 1    & i_2'\\
f_5  & g_3  & 1    & 0\\
f_5' & 1    & h_3' & i_2'\\
f_5' & 0    & h_3' & 1\\
\ema,
\\&&
(iv) F_2(i_2=i_4', i_3\neq 0,1)=
\bma
0    & 0    & 0    & 0\\
0    & 1    & 0    & i_2\neq 0,1,i_3,i_3' \\
1    & g_3\ne0,1  & h_3\ne0,1  & i_3\\
1    & g_3' & h_3  & i_2'\\
f_5\ne0,1  & g_3' & 1    & i_2\\
f_5  & g_3  & 1    & i_3'\\
f_5' & 1    & h_3' & i_2'\\
f_5' & 0    & h_3' & 1\\
\ema,
\end{eqnarray}
\begin{eqnarray}
&&
F_2(i_2=i_4', i_3=0)=
\bma
0    & 0    & 0    & 0\\
0    & 1    & 0    & i_2\neq 0,1 \\
1    & g_3\ne0,1  & h_3\ne0,1  & 0\\
1    & g_3' & h_3  & i_2'\\
f_5\ne0,1  & g_3' & 1    & i_2\\
f_5  & g_3  & 1    & 1\\
f_5' & 1    & h_3' & i_2'\\
f_5' & 0    & h_3' & 1\\
\ema,
\\&&
F_2(i_2=i_4', i_3=1)=
\bma
0    & 0    & 0    & 0\\
0    & 1    & 0    & i_2\neq 0,1\\
1    & g_3\ne0,1  & h_3\ne0,1  & 1\\
1    & g_3' & h_3  & i_2'\\
f_5\ne0,1  & g_3' & 1    & i_2\\
f_5  & g_3  & 1    & 0\\
f_5' & 1    & h_3' & i_2'\\
f_5' & 0    & h_3' & 1\\
\ema,
\\&&
F_3=
\bma
0    & 0    & 0           & 0\\
0    & 1    & 0           & i_2\neq 0,1 \\
1    & g_3\ne0,1  & h_3\ne0,1,h_4        & 0\\
1    & g_3' & h_4\ne0,1         & i_2\\
f_5\ne0,1  & g_3' & 1           & i_2'\\
f_5  & g_3  & 1           & 1\\
f_5' & 1    & h_3'        & i_2'\\
f_5' & 0    & h_4'        & 1\\
\ema,
\\&&
F_4=
\bma
0    & 0                     & 0           & 0\\
0    & 1                     & 0           & i_2\neq 0,1,i_3,i_3' \\
1    & g_3 \neq 0,1,g_4,g_4'                  & h_3\ne0,1         & i_3\ne0,1 \\
1    & g_4\neq 0,1                  & h_3         & i_3'\\
f_5\ne0,1  & g_3'                  & 1           & i_3\\
f_5  & g_4'                  & 1           & i_3'\\
f_5' & 1                     & h_3'        & i_2'\\
f_5' & 0                     & h_3'        & 1\\
\ema,
\\&&
F_5=
\bma
0               & 0           & 0           & 0\\
0               & 1           & 0           & i_2 \neq 0,1\\
1               & g_3\neq 0,1         & h_3\neq 0,1         & 1\\
1               & g_3'        & h_3         & i_2'\\
f_5\neq 0,1,f_6          & g_3         & 1           & 0\\
f_6 \neq 0,1            & g_3'        & 1           & i_2\\
f_5'            & 1           & h_3'        & i_2'\\
f_6'            & 0           & h_3'        & 1\\
\ema,
\\&& F_6=
\bma
0&0&0&0\\
0&1&h_2\neq 0,1&i_2\neq 0,i_3',i_4\\
1&g_3\neq 0,1&0&i_3\neq 0,i_4\\
1&g_3'&h_2&i_4\neq 0,1\\
f_5\ne0,1 &1&h_2'&i_3'\\
f_5&0&1&i_4'\\
f_5'&g_3&1&i_2'\\
f_5'&g_3'&h_2'&1\\
\ema,
\\&&
(i)\ F_6(i_2=1)=
\bma
0&0&0&0\\
0&1&h_2\neq 0,1&1\\
1&g_3\neq 0,1&0&i_3\neq 0,1,i_4\\
1&g_3'&h_2&i_4\neq 0,1\\
f_5\ne0,1 &1&h_2'&i_3'\\
f_5&0&1&i_4'\\
f_5'&g_3&1&0\\
f_5'&g_3'&h_2'&1\\
\ema,
\end{eqnarray}
\begin{eqnarray}
&&
F_6(i_2=1,i_3=i_4')=
\bma
0&0&0&0\\
0&1&h_2\neq 0,1&1\\
1&g_3\neq 0,1&0&i_3\\
1&g_3'&h_2&i_3'\neq 0,1\\
f_5\ne0,1 &1&h_2'&i_3'\\
f_5&0&1&i_3\\
f_5'&g_3&1&0\\
f_5'&g_3'&h_2'&1\\
\ema,
\\&& 
(ii)\ F_6(i_2=i_3)=
\bma
0&0&0&0\\
0&1&h_2\neq 0,1&i_2\ne0,1,i_4,i_4'\\
1&g_3\neq 0,1&0&i_2\\
1&g_3'&h_2&i_4\neq 0,1\\
f_5\ne0,1 &1&h_2'&i_2'\\
f_5&0&1&i_4'\\
f_5'&g_3&1&i_2'\\
f_5'&g_3'&h_2'&1\\
\ema,
\\&&
(iii)F_6(i_2=i_3')=
\bma
0&0&0&0\\
0&1&h_2\neq 0,1&i_2\neq 0,1,i_4,i_4'\\
1&g_3\neq 0,1&0&i_2'\\
1&g_3'&h_2&i_4\neq 0\\
f_5\ne0,1 &1&h_2'&i_2\\
f_5&0&1&i_4'\\
f_5'&g_3&1&i_2'\\
f_5'&g_3'&h_2'&1\\
\ema,
\\&&
F_6(i_2=i_3',i_4=1)=
\bma
0&0&0&0\\
0&1&h_2\neq 0,1&i_2\neq 0,1\\
1&g_3\neq 0,1&0&i_2'\\
1&g_3'&h_2&1\\
f_5\ne0,1 &1&h_2'&i_2\\
f_5&0&1&0\\
f_5'&g_3&1&i_2'\\
f_5'&g_3'&h_2'&1\\
\ema,
\\\label{eq:UOM7i2=i4'-summary}
&&(iv)\
F_6(i_2=i_4')=
\bma
0&0&0&0\\
0&1&h_2\neq 0,1&i_2\ne0,1,i_3,i_3'\\
1&g_3\neq 0,1&0&i_3\neq 0\\
1&g_3'&h_2&i_2'\\
f_5\ne0,1 &1&h_2'&i_3'\\
f_5&0&1&i_2\\
f_5'&g_3&1&i_2'\\
f_5'&g_3'&h_2'&1\\
\ema.
\end{eqnarray}
\end{widetext}

\section{The proof of Lemma \ref{le:t11=3} }
\label{app:proof}

In this section we prove Lemma \ref{le:t11=3}.

(i) Using $F_1$ in Appendix \ref{app:uom,f1-f6} we have
\begin{eqnarray}
\label{eq:f1-s1}
S_{11}=
\bma
0&1&0&1\\
1 & g_3\ne0,1 & h_3\ne0,1 &i_3\ne0,1,i_4\\
1 & g_3' & h_3 &i_4\ne0,1\\
f_5\ne0,1 & g_3' & 1 &i_4'\\
f_5 & g_3 & 1 &i_3'\\
f_5' & 0 & h_3' &1\\
f_5' & 1 & h_3' &0\\
\ema.
\end{eqnarray}
Let the $4\times4$ product state $\ket{x,y}\in\cT_{11}$. Lemma \ref{le:fj=property} (i) and (ii) imply that we have two cases. First,
$\ket{x}$ is orthogonal to four product vectors in $\{\ket{a_{1i},b_{1i}},i=2,...,8\}$, and 
$\ket{y}$ is orthogonal to three product vectors in $\{\ket{c_{1i},d_{1i}},i=2,...,8\}$. Second,
$\ket{x}$ is orthogonal to three product vectors in $\{\ket{a_{1i},b_{1i}},i=2,...,8\}$, and 
$\ket{y}$ is orthogonal to four product vectors in $\{\ket{c_{1i},d_{1i}},i=2,...,8\}$.

In the first case using Lemma \ref{le:fj=property} (iii) we obtain
\begin{eqnarray}
\label{eq:beta1}
\ket{0,0,0,0},
\ket{f_5',0,h_3',0},
\ket{f_5,g_3,\b_3},
\ket{f_5,g_3',\b_4}
\in\cT_{21},
\end{eqnarray}
where $\b_3,\b_4$ are two-qubit states in $\bbH_C\otimes\bbH_D$.
In the second case, we still use Lemma \ref{le:fj=property} (iv) and exclude the same states in \eqref{eq:beta1}. If $i_3=i_4'$ then we obtain
\begin{eqnarray}
\ket{\b_5,h_3,i_3},
\ket{\b_6,h_3,i_3'}
\in\cT_{21}.
\end{eqnarray}
Hence $\abs{\cT_{11}}=4$ or $6$. The latter holds if and only if $i_3=i_4'$.

(ii) Using $F_2(i_2=i_3, i_4=0)$ we have
\begin{eqnarray}
&&
S_{21}(i_2=i_3, i_4=0)
\notag\\=&&
\bma
0    & 1    & 0    & i_3\ne0,1\\
1    & g_3\ne0,1  & h_3\ne0,1  & i_3\\
1    & g_3' & h_3  & 0\\
f_5\ne0,1  & g_3' & 1    & 1\\
f_5  & g_3  & 1    & i_3'\\
f_5' & 1    & h_3' & i_3'\\
f_5' & 0    & h_3' & 1\\
\ema.
\notag
\end{eqnarray}
Similar to case (i), we obtain
$
\ket{0,0,0,0},
\ket{f_5',0,h_3',0},
\ket{f_5,g_3,0,i_3'},
\ket{f_5,g_3',\g_1},$
\\
$\ket{\g_2,h_3',i_3},
\ket{f_5',g_3,h_3,i_3'}
\in\cT_{21}(i_2=i_3, i_4=0).
$ Hence $\abs{\cT_{21}(i_2=i_3, i_4=0)}=6$.

(iii) Using $F_3$ we have
$
\ket{0,0,0,0},
\ket{f_5',0,\d_1},
\ket{f_5,g_3,\d_2},
\ket{f_5,g_3',0,i_2'}
\in\cT_{31}.
$
Further if $h_3=h_4'$ then
$
\ket{\d_3,h_3,0}
\in\cT_{31}.	
$
Hence $\abs{\cT_{31}}=4$ or $5$.

(iv) Using $F_4$ we have
$
\ket{0,0,0,0},
\ket{f_5',0,h_3',0},
\ket{\e_1,h_3,i_3},
\ket{\e_2,h_3,i_3'}
\in\cT_{41}.
$
Hence $\abs{\cT_{41}}=4$.

(v) Using $F_5$ we have
$
\ket{0,0,0,0},
\ket{\z_1,h_3',i_2'},
\ket{\z_2,0,i_2},
\ket{0,g_3',h_3,i_2'}
\in\cT_{51}.
$
If $f_5=f_6'$ then
$
\ket{f_5',g_3,\z_3},
\ket{f_5,g_3',\z_4}
\in\cT_{61}.
$
Hence $\abs{\cT_{51}}=4$ or $6$.

(vi) Using $F_6(i_2=i_3)$ we have
$
\ket{0,0,0,0},
\ket{f_5',g_3',h_2',0},
\ket{f_5',g_3,0,i_2'},
\ket{f_5,0,\eta_1},
$\\
$
\ket{\eta_2,h_2',i_2},
\ket{\eta_3,h_2,i_2'}
\in\cT_{61}.
$
Hence $\abs{\cT_{61}}=6$.

\section{The construction of six four-qubit UOMs $F_1,F_2,...,F_6$ (to be shortened greatly)}

\label{app:construct f1-f6}

We introduce a simple fact from \cite[Lemma 2]{Joh13}. It will be used in the proof of Lemma \ref{le:independent}.
\begin{lemma}
\label{le:xx'}
(i) If $\cS \subseteq (\bbC^2)^{\ox n}$ is a UPB, then for all $\ket{v}\in\cS$ and all integers $1\le j\le n$ there is another product vector $\ket{w}\in\cS$ such that $\ket{v}$ and $\ket{w}$ are orthogonal on the $j$-th subsystem or $j$-th qubit.	

(ii) The number of distinct vectors of any qubit in a UPB is an even integer.
\end{lemma}

The following result from \cite[Lemma 5]{Chen2018Multiqubit} will be used in the proof of Lemma \ref{le:44=a3b3}. 
\begin{lemma}
\label{le:UOM}
Let X=$[x_{i,j}]$$\in\cO(m,n)$ be a UOM, and $\m(x)$ the multiplicity of element $x$. 
If $p_j=\sum\mu(x)\mu(x')$, where the summation is over all pairs $\{x,x'\}$ in column j of X, then $\sum p_{j}\geq m(m-1)/2$.
\end{lemma}
We refer to the positive integer $p_j$ as the \textit{o-number} of column $j$ of $X$. It represents 
the number of all orthogonal pairs in column $j$ of $X$. We refer readers to \cite{Chen2018Multiqubit} for more details on UOMs.

We present a special partition of a positive integer into smaller positive integers. It will be also used in the proof of Lemma \ref{le:44=a3b3}. 
\begin{lemma}
\label{le:maxsum}
Suppose $p$ is the sum of $2n$ positive integers $a_1,a_2,...,a_{2n}$. Then the maximum of $a_1a_2+a_3a_4+...+a_{2n-1}a_{2n}$ is $\lc\frac{p-2n+2}{2}\rc\cdot
\lf\frac{p-2n+2}{2}\rf+n-1$. It is achievable if and only if up to the permutation of subscripts, we have
 $a_{1}=\lc\frac{p-2n+2}{2}\rc$, $a_{2}=\lf\frac{p-2n+2}{2}\rf$ and $a_i=1$ for $i>2$.
\end{lemma}
\begin{proof}
Let $N=a_1a_2+a_3a_4+...+a_{2n-1}a_{2n}$. We fix the values of $a_1,a_3,..,a_{2n-1}$ and $a_6, a_8,..,a_{2n}$, and make $a_1>a_3>...>a_{2n-1}$ by renaming the subscripts. Then we have $a_2+a_4=p-a_1-a_3-a_5-a_6-...-a_{2n}$. Since $a_1>a_3$ is given, so $a_4=1$ and $a_2=p-a_1-a_3-a_5-a_6-...-a_{2n}-1$ come to $N$ greater. When free $a_2$, $a_4$ and $a_6$, using the argument same to freeing $a_2$ and $a_4$, $N$ is greater if $a_4=a_6=1$ and $a_2=p-a_1-a_3-a_5-a_7-a_8-...-a_{2n}-2$ is satisfied . Deducing the rest by this method, one obtain that $N$ reaches the maximum when we fix $a_1$, $a_3$,...,$a_{2n-1}$ and take $a_4=a_6=...=a_{2n}=1$, $a_2=p-a_1-a_3-...-a_{2n-1}-(n-1)$. Using the similar argument to fixing $a_1$, $a_3$,...,$a_{2n-1}$, one can show that $N$ reaches the maximum when we fix $a_2$, $a_4$,...,$a_{2n}$ and take $a_3=a_5=...=a_{2n-1}=1$, $a_1=p-a_2-a_4-...-a_{2n}-(n-1)$. It is to say that, we have $a_3=a_4=...=a_{2n}=1$ and $a_1+a_2=p-(2n-2)$. So $N$  reaches the maximum when $a_1=\lc\frac{p-2n+2}{2}\rc$, $a_2=\lf\frac{p-2n+2}{2}\rf$ and $a_i=1$ for $i=3,4,...,2n$. 
\end{proof}

From now on we will study 4-qubit UPBs of size 8, and show how to find the UPBs $\cF_1,..,\cF_6$. First of all, we present the following observation by counting the data in \cite{Joh13}. 

\begin{lemma}
\label{le:size8>=3}
Let $\cT_{A:B:C:D}$=\{$\ket{f_1,g_1,h_1,i_1}$, $\ket{f_2,g_2,h_2,i_2}$,...,$\ket{f_8,g_8,h_8,i_8}$\} be a UPB of size $8$. If one of the following three conditions holds, then  $\cT_{AB:CD}$ is not a UPB in the coarse graining $\bbC^4\ox\bbC^4$.

(i) $\cT_{A:B:C:D}$ has a qubit having at least four identical vectors.

(ii) There are three subscripts $j_1,j_2,j_3$ such that $\ket{f_{j_1}}=\ket{f_{j_2}}=\ket{f_{j_3}}$ and $\ket{g_{j_1}}=\ket{g_{j_2}}=\ket{g_{j_3}}$. 

(iii) There are five distinct subscripts $j_1,j_2,...,j_5$ such that $\ket{f_{j_1}}=\ket{f_{j_2}}=\ket{f_{j_3}}$ and $\ket{g_{j_4}}$=$\ket{g_{j_5}}$.
\end{lemma}

\begin{proof}
Take $\cT_{AB:CD}$=\{$\ket{p_1,q_1}$, $\ket{p_2,q_2}$,...,$\ket{p_8,q_8}$\}, where $\ket{p_j}=\ket{f_j,g_j}$ and $\ket{q_j}=\ket{h_j,i_j}$ for $j=1,2,...,8$.

(i) By renaming the subscripts and permuting the qubits we may assume that $\ket{f_1}=\ket{f_2}=\ket{f_3}=\ket{f_4}$. Also $\ket{g_1}$, $\ket{g_2}$, $\ket{g_3}$, $\ket{g_4}$ is linearly dependent. So $\ket{p_1}$, $\ket{p_2}$, $\ket{p_3}$, $\ket{p_4}$ are linearly dependent. Then the space spanned by $\ket{p_1}$, $\ket{p_2}$, $\ket{p_3}$, $\ket{p_4}$ has dimension at most two. Therefore, there exists a $\ket{p}\in\bbC^4$ orthogonal to $\ket{p_1}$, $\ket{p_2}$, $\ket{p_3}$, $\ket{p_4}$ and $\ket{p_5}$. Moreover, there is a $\ket{q}\in\bbC^4$ orthogonal to $\ket{q_6}$, $\ket{q_7}$ and $\ket{q_8}$. So $\ket{p,q}$ is orthogonal to $\cT_{AB:CD}$. By definition $\cT_{AB:CD}$ is not a UPB in $\bbC^4\ox\bbC^4$.

(ii) We have $\ket{p_{j_1}}$=$\ket{p_{j_2}}$=$\ket{p_{j_3}}$. Therefore, the space spanned by $\ket{p_{j_1}}$, $\ket{p_{j_2}}$, $\ket{p_{j_3}}$, $\ket{p_{j_4}}$ and $\ket{p_{j_5}}$ has dimension at most three. Besides the space spanned by $\ket{q_{j_6}}$, $\ket{q_{j_7}}$ and $\ket{q_{j_8}}$ also has dimension at most three. Then, there exist $\ket{p}$, $\ket{q}$ $\in\bbC^4$ such that $\ket{p,q}$ is orthogonal to $\cT_{AB:CD}$. 

(iii) Since the space spanned by $\ket{q_7}$, $\ket{q_8}$ and $\ket{q_9}$ has dimension at most three, there is a $\ket{q}\in\bbC^4$ orthogonal to $\ket{q_7}$, $\ket{q_8}$ and $\ket{q_9}$. Then $\ket{f_{i_1}',g_{i_2}',q}$ is orthogonal to $\cT_{AB:CD}$.
\end{proof}

\begin{lemma}
\label{le:independent}
Let $\cT_{A:B:C:D}$=\{$\ket{f_1,g_1,h_1,i_1}$, $\ket{f_2,g_2,h_2,i_2}$,...,$\ket{f_8,g_8,h_8,i_8}$\} be a UPB of size $8$. If there are $\ket{f_1}=\ket{f_2}=\ket{f_3}$ and $\ket{g_2}=\ket{g_3}=\ket{g_4}$, then the remaining vectors $\ket{f_4}$, $\ket{f_5}$,...,$\ket{f_8}$ are pairwise linearly independent or $\cT_{AB:CD}$ is no longer a UPB in $\cH_{AB}:\cH_{CD}$. Similarly, the remaining vectors $\ket{g_1}$, $\ket{g_5}$, $\ket{g_6}$,...,$\ket{g_8}$ are pairwise linearly independent or $\cT_{AB:CD}$ is no longer a UPB in $\cH_{AB}:\cH_{CD}$.
\end{lemma}

\begin{proof}
First of all, we prove that either the vectors $\ket{f_4}$, $\ket{f_5}$,...,$\ket{f_8}$ are pairwise linearly independent or $\cT_{AB:CD}$ is no longer a UPB in $\cH_{AB}:\cH_{CD}$.  If the set \{$f_4$, $f_5$,..., $f_8$\} has two identical elements, then the same element has multipicity two, three, four or five. When the element has multiplicity three, four or five, $\cT_{AB:CD}$ is no longer a UPB from lemma \ref{le:size8>=3} (iii). Also a qubit has an even number of distinct elements from lemma \ref{le:xx'} (ii). So when the set \{$f_4$, $f_5$,..., $f_8$\} has an element of multiplicity two, it must contain two different elements of both multiplicity two. It is a contradiction with lemma \ref{le:size8>=3} (iii). Now we have proved it.

Using the similar argument, one may show that the vectors $\ket{g_1}$, $\ket{g_5}$, $\ket{g_6}$,...,$\ket{g_8}$ are pairwise linearly independent or $\cT_{AB:CD}$ is no longer a UPB in $\cH_{AB}:\cH_{CD}$.
\end{proof}

\begin{lemma}
\label{le:44=a3b3}
Let $\cT_{A:B:C:D}$=\{$\ket{f_1,g_1,h_1,i_1}$, $\ket{f_2,g_2,h_2,i_2}$,...,$\ket{f_8,g_8,h_8,i_8}$\} be a 4-qubit UPB of size $8$. If the first and second qubit respectively have three identical vectors, then $\cT_{A:B:C:D}$ is not a UPB of size $8$ in $\cH_{AB}\ox\cH_{CD}$.
\end{lemma}
\begin{proof}
Suppose $\ket{a_1}$, $\ket{a_2}$, $\ket{a_3}$ are three identical vectors among $\ket{f_1}$, $\ket{f_2}$,...,$\ket{f_8}$ and $\ket{b_1}$, $\ket{b_2}$, $\ket{b_3}$ are three identical vectors among $\ket{g_1}$, $\ket{g_2}$,...,$\ket{g_8}$. If $\ket{a_1}$, $\ket{a_2}$, $\ket{a_3}$, $\ket{b_1}$, $\ket{b_2}$, $\ket{b_3}$ are in three, five or six distinct product vectors of $\cT_{A:B:C:D}$, then $\cT_{AB:CD}$ is not a UPB from Lemma \ref{le:size8>=3} (ii) and (iii).

We only need to investigate the case that $\ket{a_1}$, $\ket{a_2}$, $\ket{a_3}$, $\ket{b_1}$, $\ket{b_2}$, $\ket{b_3}$ are in four distinct product vectors. Denote by $U$ the UOM over the UPB $\cT_{A:B:C:D}$. Moreover, $p_j$ are o-number of column $j$ of the $U$ for $j=1,2,3,4$. Then $U$ has $p_1=5$ and $p_2=5$ in condition of Lemma \ref{le:independent}. Also $p_1+p_2+p_3+p_4\geq8(8-1)/2$ from Lemma \ref{le:UOM} (vi). So $p_3+p_4\geq18$. We have $p_3$, $p_4$ $\leq (\frac{8-2n+2}{2})^2=n^2-9n+24$ from Lemma \ref{le:maxsum} and $p=8$ in condition of Lemma \ref{le:maxsum}.  Then the possible value of $n$ is $1,2,3,4$. If $n=1$, then $a_1=a_2=\frac{8-2\times 1+2}{2}=4$ in Lemma \ref{le:maxsum}. That is, $\cT_{A:B:C:D}$ has a qubit having four identical vectors. Then $\cT_{A:B:C:D}$ is not a UPB from Lemma \ref{le:size8>=3} (i). If $n=3$ or $4$, then $p_3, p_4\leq 6$ or $4$. It makes a contradiction with $p_3+p_4\geq18$. If $n=2$, then $p_3, p_4\leq 10$. So the case $n=2$ is the only case that satisfies the condition $p_3+p_4\geq18$. To satisfy $p_3+p_4\geq18$, one of them must be $3\times 3+1\times 1$ and the other one could be either of $3\times 3+1\times 1$, $3\times 2+2\times 1$, $2\times 2+ 2\times 2$. There is no harm in supposing $p_3=3\times 3+1\times 1$. Then $p_4=3\times 3+1\times 1$, $3\times 2+2\times 1$ or $2\times 2+ 2\times 2$. We can obtain $h_{i_1}=h_{i_2}=h_{i_3}$ and $i_{i_4}=i_{i_5}$ for five distinct descripts $i_1$, $i_2$, $i_3$, $i_4$, $i_5$ $\in \{1, 2, 3, 4, 5, 6, 7, 8\}$. So $\cT_{A:B:C:D}$ is not a UPB from Lemma \ref{le:size8>=3} (iii).
\end{proof}

\begin{lemma}
\label{le:3f3h}
Let $\cT_{A:B:C:D}$=\{$\ket{f_1,g_1,h_1,i_1}$, $\ket{f_2,g_2,h_2,i_2}$,...,$\ket{f_8,g_8,h_8,i_8}$\} be a 4-qubit UPB of size $8$. Then $\cT_{AB:CD}$ is not a UPB of size $8$ in $\cH_{AB}\ox\cH_{CD}$ when one of the following three conditions is satisfied. 

(i) There are three subscripts $j_1,j_2,j_3$ such that $\ket{f_{j_1}}=\ket{f_{j_2}}=\ket{f_{j_3}}$ and $\ket{h_{j_1}}=\ket{h_{j_2}}=\ket{h_{j_3}}$.

(ii) There are three subscripts $j_1,j_2,j_3$ such that $\ket{f_{j_1}}=\ket{f_{j_2}}=\ket{f_{j_3}}$, $\ket{g_{j_1}}=\ket{g_{j_2}}$ and $\ket{h_{j_1}}=\ket{h_{j_2}}$. 

(iii) There are three subscripts $j_1,j_2,j_3$ such that $\ket{f_{j_1}}=\ket{f_{j_2}}$, $\ket{g_{j_1}}=\ket{g_{j_2}}$ and $\ket{h_{j_1}}=\ket{h_{j_2}}=\ket{h_{j_3}}$. 
\end{lemma}
\begin{proof}
(i)  We prove the assertion by contradiction. Suppose $\cT_{AB:CD}$ is a UPB of size $8$ in $\cH_{AB}\ox\cH_{CD}$.
Up to the equivalence, we can assume $j_1=1$, $j_2=2$ and $j_3=3$, and $f_1=f_2=f_3=h_1=h_2=h_3=0$. We express the UOM of $\cT_{A:B:C:D}$ as
\begin{eqnarray}
U_1=
\bma
0 & g_1 & 0 & i_1\\
0 & g_2 & 0 & i_2\\
0 & g_3 & 0 & i_3\\
f_4 & g_4 & h_4 & i_4\\
f_5 & g_5 & h_5 & i_5\\
f_6 & g_6 & h_6 & i_6\\
f_7 & g_7 & h_7 & i_7\\
f_8 & g_8 & h_8 & i_8\\
\ema.	
\end{eqnarray}
Since the first three rows of $U_1$ correspond to three orthogonal product vectors, we obtain that $\ket{g_1,i_1},\ket{g_2,i_2},\ket{g_3,i_3}$ are orthogonal. Up to equivalence we may assume that 
$g_1=g_2=0$ and $g_3=1$. Since $\cT_{AB:CD}$ is a UPB of size $8$ in $\cH_{AB}\ox\cH_{CD}$,
Lemma \ref{le:size8>=3} (iii) shows that $g_4,g_5,...,g_8$ are distinct. In addition, we have $g_4,g_5,...,g_8\neq 0$ from Lemma \ref{le:44=a3b3}. So one of them must be $1$ by Lemma \ref{le:xx'} (ii).
So $U_1$ is equivalent to 
\begin{eqnarray}
\label{eq:U1+}
\bma
0 & 0 & 0 & i_1\\
0 & 0 & 0 & i_1'\\
0 & 1 & 0 & i_3\\
f_4 & 1 & h_4 & i_4\\
f_5 & g_5 & h_5 & i_5\\
f_6 & g_5' & h_6 & i_6\\
f_7 & g_7 & h_7 & i_7\\
f_8 & g_7' & h_8 & i_8\\
\ema.
\end{eqnarray}
We have $g_5,g_7\neq 0,1$ from the discussion in the paragraph above \eqref{eq:U1+}. Since the first and second row vectors are orthogonal to the last four row vectors of $U_2$, we obtain that $\ket{0,0}\in\cH_A\ox\cH_C$ is orthogonal to $\ket{f_j,h_j}$ for $j=5,6,7,8$. Since $\cT_{AB:CD}$ is a UPB of size $8$, Lemma \ref{le:size8>=3} (iii) shows that $f_5,f_6,f_7,f_8$ contain exactly two $1$'s, and so do $h_5,h_6,h_7,h_8$. So the matrix in \eqref{eq:U1+} is equivalent to 
\begin{eqnarray}
U_{11}=
\bma
0 & 0 & 0 & i_1\\
0 & 0 & 0 & i_1'\\
0 & 1 & 0 & i_3\\
f_4 & 1 & h_4 & i_4\\
1 & g_5 & h_5 & i_5\\
1 & g_5' & h_6 & i_6\\
f_7 & g_7 & 1 & i_7\\
f_8 & g_7' & 1 & i_8\\
\ema
\text{\quad or\quad}
U_{12}=
\bma
0 & 0 & 0 & i_1\\
0 & 0 & 0 & i_1'\\
0 & 1 & 0 & i_3\\
f_4 & 1 & h_4 & i_4\\
1 & g_5 & h_5 & i_5\\
f_6 & g_5' & 1 & i_6\\
1 & g_7 & h_7 & i_7\\
f_8 & g_7' & 1 & i_8\\
\ema,
\end{eqnarray}
where $g_5\neq g_7,g_7'$. For $U_{11}$, Lemma \ref{le:size8>=3} shows that $h_5,f_7,f_8\ne0$. Since row $5$ is orthogonal to row $7$ and $8$, we have $i_5'=i_7=i_8$. So column $3$ and $4$ of $U_{11}$ shows a contradiction with Lemma \ref{le:size8>=3} (iii) and the fact that $\cT_{AB:CD}$ is a UPB of size $8$. 

On the other hand for $U_{12}$, similar to the above argument for $U_{11}$ one can show that $i_8=i_5'$, $i_7=i_6'$, $f_8=f_6'$ and $h_7=h_5'$. Then row $4$ of $U_{12}$ is not orthogonal to all four bottom row vectors of $U_{12}$. It is a contradiction with the fact that $U_{12}$ is a UOM. We have proven that $\cT_{AB:CD}$ is not a UPB of size $8$ in $\cH_{AB}\ox\cH_{CD}$.

(ii) We prove the assertion by contradiction. Suppose $\cT_{AB:CD}$ is a UPB of size $8$ in $\cH_{AB}\ox\cH_{CD}$.
Up to the equivalence, we can assume $j_1=1$, $j_2=2$ and $j_3=3$ and $f_1=f_2=f_3=g_1=g_2=h_1=h_2=0$. We express the UOM of $\cT_{A:B:C:D}$ as
\begin{eqnarray}
\label{eq:3+2+2}
U_1=
\bma
0 & 0 & 0 & 0\\
0 & 0 & 0 & 1\\
0 & g_3 & h_3 & i_3\\
f_4 & g_4 & h_4 & i_4\\
f_5 & g_5 & h_5 & i_5\\
f_6 & g_6 & h_6 & i_6\\
f_7 & g_7 & h_7 & i_7\\
f_8 & g_8 & h_8 & i_8\\
\ema.
\end{eqnarray}
We claim that there are only two cases $U_{11}$ and $U_{12}$ in \eqref{eq:startclass2}, where $f_5$ may be $f_7$ or $f_7'$ and $f_5, f_6, f_7\neq 0,1$ is satisfied.
\begin{eqnarray}
\label{eq:startclass2}
U_{11}=
\bma
0 & 0 & 0 & 0\\
0 & 0 & 0 & 1\\
0 & g_3 & h_3 & i_3\\
1 & g_4 & h_4 & i_4\\
f_5 & g_5 & h_5 & i_5\\
f_5' & g_6 & h_6 & i_6\\
f_7 & g_7 & h_7 & i_7\\
f_7' & g_8 & h_8 & i_8\\
\ema,
\text{\quad}
U_{12}=
\bma
0 & 0 & 0 & 0\\
0 & 0 & 0 & 1\\
0 & g_3 & h_3 & i_3\\
1 & g_4 & h_4 & i_4\\
1 & g_5 & h_5 & i_5\\
f_6 & g_6 & h_6 & i_6\\
f_6 & g_7 & h_7 & i_7\\
f_6' & g_8 & h_8 & i_8\\
\ema.
\end{eqnarray}
We can obtain that $k$ of $f_4, f_5, f_6, f_7, f_8$ of $U_1$ equal to $1$ from Lemma\ref{le:xx'}(i), where $1\leq k\leq 5$ and $k$ is a positive integer. Moreover, we have $f_4, f_5, f_6, f_7, f_8 \neq 0$ and $1\leq k\leq 3$ from Lemma \ref{le:size8>=3}(i). However, if $k=3$, then column 1, 2 of $U_1$ make a contradiction with the fact that $U_1$ is a UOM and Lemma \ref{le:size8>=3}(iii). Then we have $k=1,2$. Namely, up to equivalence we obtain two cases, $f_4=1,f_5,f_6,f_7,f_8\neq 1$ and $f_4=f_5=1,f_6,f_7,f_8\neq 1$. For $f_4=1,f_5,f_6,f_7,f_8\neq 1$, at most two of $f_5,f_6,f_7,f_8$ are the same from \ref{le:size8>=3}(iii). Moreover, from Lemma \ref{le:xx'}(i) we can obtain $f_6=f_5',f_8=f_7'$ up to equivalent, where $f_5$ may be $f_7$ or $f_7'$. So we have proved the claim in the line above \eqref{eq:startclass2}.

For $U_{11}$ in \eqref{eq:startclass2}, we claim that there are two cases, $U_{111}$ and $U_{112}$. We can obtain $f_5,f_5',f_7,f_7'\neq 1$ from $f_5,f_7\neq 0,1$ in the line above \eqref{eq:startclass2}.
Since row 1, 2 are othogonal to row 3, 5, 6, 7 and 8 of $U_{11}$, we obtain that $\ket{0,0}\in\cH_B\ox\cH_C$ is orthogonal to $g_j,h_j$ for $j=3,5,6,7,8$. First, one can show at most two of $g_3,g_5,g_6,g_7,g_8$ are $1$'s. Otherwise, row 1 and 2 of $U_{11}$ is a contradiction with Lemma \ref{le:size8>=3}(iii) and the fact that $U_{11}$ is a UOM by the assumption $U_1$ in \eqref{eq:3+2+2} is a UOM. Second, one can show at most three of $h_3,h_5,h_6,h_7,h_8$ are $1$'s from Lemma \ref{le:size8>=3}(i). Then we have shown that two of $g_3,g_5,g_6,g_7,g_8$ are $1$'s and three of $h_3,h_5,h_6,h_7,h_8$ are $1$'s in $U_{11}$. If $g_3=1$, up to equivalence we can assume $g_5=1$. We directly obtain $h_6=h_7=h_8=1$. Then $U_{11}$ becomes $U_{111}$. On the other hand for $g_3\neq 1$, that is $h_3=1$, up to equivalence we can assume $g_5=g_6=1$. We directly obtain $h_7=h_8=1$. Then $U_{11}$ becomes $U_{112}$. Now we have proved the claim at the beginning of this paragraph. 

There is $\ket{1,0,0,i_4'}\in\cH_A\ox\cH_B\ox\cH_C\ox\cH_D$ orthogonal to all row vectors of $U_{111}$ and $U_{112}$. It is a contradiction with the fact $U_{111}$ and $U_{112}$ are UOMs of size 8 in $\cH_{CD}\ox\cH_{CD}$ by the assumption $U_1$ in \eqref{eq:3+2+2} is a UOM.
\begin{eqnarray}
U_{111}=
\bma
0 & 0 & 0 & 0\\
0 & 0 & 0 & 1\\
0 & 1 & h_3 & i_3\\
1 & g_4 & h_4 & i_4\\
f_5 & 1 & h_5 & i_5\\
f_5' & g_6 & 1 & i_6\\
f_7 & g_7 & 1 & i_7\\
f_7' & g_8 & 1 & i_8\\
\ema,
\text{\quad}
U_{112}=
\bma
0 & 0 & 0 & 0\\
0 & 0 & 0 & 1\\
0 & g_3 & 1 & i_3\\
1 & g_4 & h_4 & i_4\\
f_5 & 1 & h_5 & i_5\\
f_5' & 1 & h_6 & i_6\\
f_7 & g_7 & 1 & i_7\\
f_7' & g_8 & 1 & i_8\\
\ema.
\end{eqnarray}

For $U_{12}$ in \eqref{eq:startclass2}, we claim that there are four cases $U_{121}$, $U_{122}$, $U_{123}$ and $U_{124}$ in \eqref{eq:2+2of3+2+2}. Since row 1, 2 are othogonal to row 3, 6, 7 and 8 of $U_{11}$, we obtain that $\ket{0,0}\in\cH_B\ox\cH_C$ is orthogonal to $\ket{g_j,h_j}$ for $j=3,6,7,8$. First, one can show at most two of $g_3,g_6,g_7,g_8$ are $1$'s. Otherwise,  row 1 and 2 of $U_{11}$ is a contradiction with Lemma \ref{le:size8>=3}(iii) and the fact that $U_{12}$ is a UOM by the assumption that $U_1$ in \eqref{eq:3+2+2} is a UOM. Second, one can show at most three of $h_3,h_6,h_7,h_8$ are $1$'s from Lemma \ref{le:size8>=3}(i). If three of $h_3,h_6,h_7,h_8$ are $1$'s, then there exists $\ket{0,1,0,i_3'}$ orthogonal to $U_{12}$ for $g_3=1$ and there exists $\ket{0,1,0,i_j}$ orthogonal to $U_{12}$ for $h_3=1$ and $g_j=1$ for $j=6$, or $7$, or $8$. It is a contradiction with the definition of UPB and the fact $U_{12}$ is a UOM by the assumption $U_1$ in \eqref{eq:3+2+2} is a UOM. Then we have shown that two of $g_3,g_6,g_7,g_8$ are $1$'s and two of $h_3,h_6,h_7,h_8$ are $1$'s in $U_{12}$. For $g_3=1$, we have two cases, $g_6=1$ and $g_8=1$. In case one one can directly obtain $h_7=h_8=1$. Then $U_{12}$ becomes $U_{121}$. In case two one can directly obtain $h_6=h_7=1$. Then $U_{12}$ becomes $U_{122}$. On the other hand for $g_3\neq 1$, that is $h_3=1$, we also have two cases, $g_6=g_7=1$ and $g_6=g_8=1$. In case one one can directly obtain $h_8=1$. Then $U_{12}$ becomes $U_{123}$. In case two one can directly obtain $h_7=1$. Then $U_{12}$ becomes $U_{124}$. Now we have proved the claim at the beginning of this paragraph. 
\begin{widetext}
\begin{eqnarray}
\label{eq:2+2of3+2+2}
U_{121}=
\bma
0 & 0 & 0 & 0\\
0 & 0 & 0 & 1\\
0 & 1 & h_3 & i_3\\
1 & g_4 & h_4 & i_4\\
1 & g_5 & h_5 & i_5\\
f_6 & 1 & h_6 & i_6\\
f_6 & g_7 & 1 & i_7\\
f_6' & g_8 & 1 & i_8\\
\ema,
\text{\quad}
U_{122}=
\bma
0 & 0 & 0 & 0\\
0 & 0 & 0 & 1\\
0 & 1 & h_3 & i_3\\
1 & g_4 & h_4 & i_4\\
1 & g_5 & h_5 & i_5\\
f_6 & g_6 & 1 & i_6\\
f_6 & g_7 & 1 & i_7\\
f_6' & 1 & h_8 & i_8\\
\ema,
\text{\quad}
U_{123}=
\bma
0 & 0 & 0 & 0\\
0 & 0 & 0 & 1\\
0 & g_3 & 1 & i_3\\
1 & g_4 & h_4 & i_4\\
1 & g_5 & h_5 & i_5\\
f_6 & 1 & h_6 & i_6\\
f_6 & 1 & h_7 & i_7\\
f_6' & g_8 & 1 & i_8\\
\ema,
\text{\quad}
U_{124}=
\bma
0 & 0 & 0 & 0\\
0 & 0 & 0 & 1\\
0 & g_3 & 1 & i_3\\
1 & g_4 & h_4 & i_4\\
1 & g_5 & h_5 & i_5\\
f_6 & 1 & h_6 & i_6\\
f_6 & g_7 & 1 & i_7\\
f_6' & 1 & h_8 & i_8\\
\ema.
\end{eqnarray}
\end{widetext}

In the following we show that neither of $U_{121}$, $U_{122}$, $U_{123}$, $U_{124}$ is a UOM in $\cH_{AB}\ox\cH_{CD}$. This will prove the claim of (ii).

For $U_{121}$ in \eqref{eq:2+2of3+2+2}, we have $f_6, f_6'\neq 1$ from the line above \eqref{eq:startclass2}. We obtain $g_7, g_8\neq 0$ from Lemma \ref{le:44=a3b3} and $h_3\neq 0$ from Lemma \ref{le:3f3h} (i). Since row 3 is orthogonal to row 7, 8, we can obtain $i_7=i_8=i_3'$. Then $U_{121}$ becomes $U_{1211}$ in \eqref{eq:12next2+2of3+2+2}. So there exists $\ket{0,0,1,i_3}\in\cH_A\ox\cH_B\ox\cH_C\ox\cH_D$ orthogonal to all row vectors of $U_{1211}$. It shows a contradiction with the definition of UOM and the fact that $U_{1211}$ is a UOM by the assumption $U_1$ in \eqref{eq:3+2+2} is a UOM.

For $U_{122}$ in \eqref{eq:2+2of3+2+2}, we have $f_6\neq 1$ from the line above \eqref{eq:startclass2}. We obtain $g_6, g_7\neq 0$ from Lemma \ref{le:44=a3b3} and $h_3\neq 0$ from Lemma \ref{le:3f3h} (i). Since row 3 is orthogonal to row 6, 7, we can obtain $i_6=i_7=i_3'$. Then $U_{122}$ becomes $U_{1221}$ in \eqref{eq:12next2+2of3+2+2}. So there exists $\ket{0,0,1,i_3}\in\cH_A\ox\cH_B\ox\cH_C\ox\cH_D$ orthogonal to all row vectors of $U_{1221}$. It shows a contradiction with the definition of UPB and the fact that $U_{1221}$ is a UOM.
\begin{eqnarray}
\label{eq:12next2+2of3+2+2}
U_{1211}=
\bma
0 & 0 & 0 & 0\\
0 & 0 & 0 & 1\\
0 & 1 & h_3 & i_3\\
1 & g_4 & h_4 & i_4\\
1 & g_5 & h_5 & i_5\\
f_6 & 1 & h_6 & i_6\\
f_6 & g_7 & 1 & i_3'\\
f_6' & g_8 & 1 & i_3'\\
\ema,
\text{\quad}
U_{1221}=
\bma
0 & 0 & 0 & 0\\
0 & 0 & 0 & 1\\
0 & 1 & h_3 & i_3\\
1 & g_4 & h_4 & i_4\\
1 & g_5 & h_5 & i_5\\
f_6 & g_6 & 1 & i_3'\\
f_6 & g_7 & 1 & i_3'\\
f_6' & 1 & h_8 & i_8\\
\ema.
\end{eqnarray}

For $U_{123}$ in \eqref{eq:2+2of3+2+2}, we have $f_6\neq 0,1$ from the line above \eqref{eq:startclass2}. We obtain $g_3,g_4,g_5\neq 0$ from Lemma \ref{le:44=a3b3}. In the following we show $h_4,h_5,h_7\neq 1$. First, one can obtain $h_4\neq 1$. Otherwise, we have $\ket{1,0,0,i_5'}\in\cH_A\ox\cH_B\ox\cH_C\ox\cH_D$ is orthogonal to all row vectors of $U_{123}$. Second, one can obtain $h_5\neq 1$. Otherwise, we have $\ket{1,0,0,i_4'}\in\cH_A\ox\cH_B\ox\cH_C\ox\cH_D$ is orthogonal to all row vectors of $U_{123}$. Lastly, one can obtain $h_7\neq 1$. Otherwise, we have $\ket{0,1,0,i_6'}\in\cH_A\ox\cH_B\ox\cH_C\ox\cH_D$ is orthogonal to all row vectors of $U_{123}$. Then we have proved
\begin{equation}
\label{eq:condition1}
f_6,f_6'\neq 0,1,\ g_3,g_4,g_5\neq 0,1\ h_4,h_5,h_7\neq 1. 
\end{equation}

We claim that $U_{123}$ in \eqref{eq:2+2of3+2+2} has two cases $U_{1231}$ and $U_{1232}$ in \eqref{eq:3next2+2of3+2+2}. Since row 3 is orthogonal to row 6, 7, we can obtain $\ket{1,i_3}$ is orthogonal to $\ket{h_6,i_6}$ and $\ket{h_7,i_7}$ from $f_6\neq 1,g_3\neq 0$ by \eqref{eq:condition1}. One can show that $h_6,h_7$ are not equal to $0$ at the same time. Otherwise, column 3 of $U_{123}$ shows a contradiction with Lemma \ref{le:size8>=3} (i) and the fact $U_{123}$ is a UOM by the assumption $U_1$ in \eqref{eq:3+2+2} is a UOM. Then we have $h_6=0,h_7\neq 0$ or $h_6\neq 0,h_7=0$ or $h_6,h_7\neq 0$. Since $f_6=f_7,g_6=g_7$ and $i_6,i_7$ is undetermined, one can obtain the first two cases $h_6=0,h_7\neq 0 $ and $h_6\neq 0,h_7=0$ are equivalent. Up to equivalence, we have two cases $h_6=0,h_7\neq 0$ or $h_6,h_7\neq 0$ for $U_{123}$. 
For $h_6=0,h_7\neq 0$ in $U_{123}$, since $\ket{1,i_3}$ is orthogonal to $\ket{h_7,i_7}$ from the second line in this paragraph, we obtain $i_7=i_3'$. Since $h_7\neq 1$ by \eqref{eq:condition1} and row 5 is orthogonal to row 6, we have $i_6=i_7'=i_3$. Since $f_6\neq 0,g_4,g_5\neq 0, h_4,h_5\neq 1$ by \eqref{eq:condition1} and row 6 is orthogonal to row 4, 5, we have $i_4=i_5=i_6'=i_3'$. Then $U_{123}$ becomes $U_{1231}$ in \eqref{eq:3next2+2of3+2+2}. 
For $h_6,h_7\neq 0$ in $U_{123}$, since $\ket{1,i_3}$ is orthogonal to $\ket{h_6,i_6}$ and $\ket{h_7,i_7}$ from the second line in this paragraph, we can obtain $i_6=i_7=i_3'$. Then $U_{123}$ becomes $U_{1232}$ in \eqref{eq:3next2+2of3+2+2}. We have proved the claim at the beginning of this paragraph.
\begin{widetext}
\begin{eqnarray}
\label{eq:3next2+2of3+2+2}
U_{1231}=
\bma
0 & 0 & 0 & 0\\
0 & 0 & 0 & 1\\
0 & g_3 & 1 & i_3\\
1 & g_4 & h_4 & i_3'\\
1 & g_5 & h_5 & i_3'\\
f_6 & 1 & 0 & i_3\\
f_6 & 1 & h_7\neq 0 & i_3'\\
f_6' & g_8 & 1 & i_8\\
\ema,
\text{\quad}
U_{1232}=
\bma
0 & 0 & 0 & 0\\
0 & 0 & 0 & 1\\
0 & g_3 & 1 & i_3\\
1 & g_4 & h_4 & i_4\\
1 & g_5 & h_5 & i_5\\
f_6 & 1 & h_6\neq 0 & i_3'\\
f_6 & 1 & h_7\neq 0 & i_3'\\
f_6' & g_8 & 1 & i_8\\
\ema.
\end{eqnarray}
\end{widetext}

For $U_{1231}$ in \eqref{eq:3next2+2of3+2+2}, there exists $\ket{1,0,0,i_3}\in\cH_A\ox\cH_B\ox\cH_C\ox\cH_D$ is orthogonal to all row vectors of $U_{1231}$. It shows a contradiction with the definition of UOM and the fact $U_{1231}$ is a UOM in $\cH_A\ox\cH_B\ox\cH_C\ox\cH_D$ by the assumption $U_1$ in \eqref{eq:3+2+2} is a UOM. That is, $U_{1231}$ is not a UOM in $\cH_{AB}\ox\cH_{CD}$.

For $U_{1232}$ in \eqref{eq:3next2+2of3+2+2}, there exists $\ket{0,1,0,i_3}\in\cH_A\ox\cH_B\ox\cH_C\ox\cH_D$ is orthogonal to all row vectors of $U_{1232}$. It shows a contradiction with the definition of UOM and the fact $U_{1232}$ is a UOM in $\cH_A\ox\cH_B\ox\cH_C\ox\cH_D$ by the assumption $U_1$ in \eqref{eq:3+2+2} is a UOM. That is, $U_{1232}$ is not a UOM in $\cH_{AB}\ox\cH_{CD}$.

Therefore, $U_{123}$ in \eqref{eq:3next2+2of3+2+2} is not a UOM in $\cH_{AB}\ox\cH_{CD}$.

For $U_{124}$ in \eqref{eq:2+2of3+2+2}, we have $f_6,f_6'\neq 0,1$ from the line above \eqref{eq:startclass2}. We obtain $g_3,g_4,g_5\neq 0$ from Lemma \ref{le:44=a3b3}. In the following we show that $h_4,h_5\neq 1$. First, one can obtain $h_4\neq 1$. Otherwise, we have $\ket{1,0,0,i_5'}\in\cH_A\ox\cH_B\ox\cH_C\ox\cH_D$ is orthogonal to all row vectors of $U_{124}$. Second, one can obtain $h_5\neq 1$. Otherwise, we have $\ket{1,0,0,i_4'}\in\cH_A\ox\cH_B\ox\cH_C\ox\cH_D$ is orthogonal to all row vectors of $U_{124}$. Then we have proved 
\begin{equation}
\label{eq:condition2}
f_6,f_6'\neq 0,1,\ g_3,g_4,g_5\neq 0,\ h_4,h_5\neq 1. 
\end{equation}

We claim that $U_{124}$ in \eqref{eq:2+2of3+2+2} has three cases $U_{1241}$, $U_{1242}$ and $U_{1243}$ in \eqref{eq:4next2+2of3+2+2}. Since row 3 is orthogonal to row 6, 8, we can obtain $\ket{1,i_3}$ is orthogonal to $\ket{h_6,i_6}$ and $\ket{h_8,i_8}$ from $f_6,f_6'\neq 1,g_3\neq 0$ by \eqref{eq:condition2}. One can show that $h_6,h_8$ are not equal to $0$ at the same time. Otherwise, column 3 of $U_{124}$ shows a contradiction with Lemma \ref{le:size8>=3} (i) and the fact $U_{124}$ is a UOM by the assumption $U_1$ in \eqref{eq:3+2+2} is a UOM. Then we have three cases $h_6=0,h_8\neq 0$ or $h_6\neq 0,h_8=0$ or $h_6,h_8\neq 0$ for $U_{124}$. 
For $h_6=0,h_8\neq 0$ in $U_{124}$, since $f_6\neq 0$, $g_4,g_5\neq 0$, $h_4,h_5\neq 1$ by \eqref{eq:condition2} and row 6 is orthogonal to row 4, 5, we have $i_4=i_5=i_6'$. Then $U_{124}$ becomes $U_{1241}$ in \eqref{eq:4next2+2of3+2+2}. 
For $h_6\neq 0,h_8=0$ in $U_{124}$, since $f_6'\neq 0$, $g_4,g_5\neq 0$, $h_4,h_5\neq 1$ by \eqref{eq:condition2} and row 8 is orthogonal to row 4, 5, we have $i_4=i_5=i_8'$. Then $U_{124}$ becomes $U_{1242}$ in \eqref{eq:4next2+2of3+2+2}. 
For $h_6,h_8\neq 0$ in $U_{124}$, since $\ket{1,i_3}$ is orthogonal to $\ket{h_6,i_6}$ and $\ket{h_7,i_7}$ from the second line in this paragraph, we can obtain $i_6=i_8=i_3'$. Then $U_{124}$ becomes $U_{1243}$ in \eqref{eq:4next2+2of3+2+2}. We have proved the claim at the beginning of this paragraph.
\begin{widetext}
\begin{eqnarray}
\label{eq:4next2+2of3+2+2}
U_{1241}=
\bma
0 & 0 & 0 & 0\\
0 & 0 & 0 & 1\\
0 & g_3 & 1 & i_3\\
1 & g_4 & h_4 & i_6'\\
1 & g_5 & h_5 & i_6'\\
f_6 & 1 & 0 & i_6\\
f_6 & g_7 & 1 & i_7\\
f_6' & 1 & h_8\neq 0 & i_8\\
\ema,
\text{\quad}
U_{1242}=
\bma
0 & 0 & 0 & 0\\
0 & 0 & 0 & 1\\
0 & g_3 & 1 & i_3\\
1 & g_4 & h_4 & i_8'\\
1 & g_5 & h_5 & i_8'\\
f_6 & 1 & h_6\neq 0 & i_6\\
f_6 & g_7 & 1 & i_7\\
f_6' & 1 & 0 & i_8\\
\ema,
\text{\quad}
U_{1243}=
\bma
0 & 0 & 0 & 0\\
0 & 0 & 0 & 1\\
0 & g_3 & 1 & i_3\\
1 & g_4 & h_4 & i_4\\
1 & g_5 & h_5 & i_5\\
f_6 & 1 & h_6\neq 0 & i_3'\\
f_6 & g_7 & 1 & i_7\\
f_6' & 1 & h_8\neq 0 & i_3'\\
\ema.
\end{eqnarray}
\end{widetext}

For $U_{1241}$ in \eqref{eq:4next2+2of3+2+2}, there exists $\ket{1,0,0,i_6}\in\cH_A\ox\cH_B\ox\cH_C\ox\cH_D$ is orthogonal to all row vectors of $U_{1241}$. It shows a contradiction with the definition of UOM and the fact $U_{1241}$ is a UOM in $\cH_A\ox\cH_B\ox\cH_C\ox\cH_D$ by the assumption $U_1$ in \eqref{eq:3+2+2} is a UOM. That is, $U_{1241}$ is not a UOM in $\cH_{AB}\ox\cH_{CD}$.

For $U_{1242}$ in \eqref{eq:4next2+2of3+2+2}, there exists $\ket{1,0,0,i_8}\in\cH_A\ox\cH_B\ox\cH_C\ox\cH_D$ is orthogonal to all row vectors of $U_{1242}$. It shows a contradiction with the definition of UOM and the fact $U_{1242}$ is a UOM in $\cH_A\ox\cH_B\ox\cH_C\ox\cH_D$ by the assumption $U_1$ in \eqref{eq:3+2+2} is a UOM. That is, $U_{1242}$ is not a UOM in $\cH_{AB}\ox\cH_{CD}$.

For $U_{1243}$ in \eqref{eq:4next2+2of3+2+2}, there exists $\ket{0,1,0,i_3}\in\cH_A\ox\cH_B\ox\cH_C\ox\cH_D$ is orthogonal to all row vectors of $U_{1243}$. It shows a contradiction with the definition of UOM and the fact $U_{1243}$ is a UOM in $\cH_A\ox\cH_B\ox\cH_C\ox\cH_D$ by the assumption $U_1$ in \eqref{eq:3+2+2} is a UOM. That is, $U_{1243}$ is not a UOM in $\cH_{AB}\ox\cH_{CD}$.

So $U_{124}$ in \eqref{eq:3next2+2of3+2+2} is not a UOM in $\cH_{AB}\ox\cH_{CD}$.

We haved proved the claim below \eqref{eq:2+2of3+2+2}.

Therefore, we have proved that $\cT_{AB:CD}$ is not a UPB of size 8 in $\cH_{AB}\ox\cH_{CD}$.

(iii)  We prove the assertion by contradiction. Suppose $\cT_{AB:CD}$ is a UPB of size $8$ in $\cH_{AB}\ox\cH_{CD}$.
Up to the equivalence, we can assume $j_1=1$, $j_2=2$ and $j_3=3$ and $f_1=f_2=g_1=g_2=h_1=h_2=h_3=0$. We express the UOM of $\cT_{A:B:C:D}$ as
\begin{eqnarray}
U_{1}=
\bma
0 & 0 & 0 & 0\\
0 & 0 & 0 & 1\\
f_3 & g_3 & 0 & i_3\\
f_4 & g_4 & h_4 & i_4\\
f_5 & g_5 & h_5 & i_5\\
f_6 & g_6 & h_6 & i_6\\
f_7 & g_7 & h_7 & i_7\\
f_8 & g_8 & h_8 & i_8\\
\ema. 
\end{eqnarray}
We claim that there are only two cases $U_{11}$ and $U_{12}$, where $c_1$, $c_2$, $c_3$, $c_4$, $c_5$ are not $0$'s or $1$'s.
\begin{eqnarray}
\label{eq:startclass3}
U_{11}=
\bma
0 & 0 & 0 & 0\\
0 & 0 & 0 & 1\\
f_3 & g_3 & 0 & i_3\\
f_4 & g_4 & 1 & i_4\\
f_5 & g_5 & c_1 & i_5\\
f_6 & g_6 & c_2 & i_6\\
f_7 & g_7 & c_3 & i_7\\
f_8 & g_8 & c_4 & i_8\\
\ema, 
\text{\quad}
U_{12}=
\bma
0 & 0 & 0 & 0\\
0 & 0 & 0 & 1\\
f_3 & g_3 & 0 & i_3\\
f_4 & g_4 & 1 & i_4\\
f_5 & g_5 & 1 & i_5\\
f_6 & g_6 & c_5 & i_6\\
f_7 & g_7 & c_5 & i_7\\
f_8 & g_8 & c_5' & i_8\\
\ema.
\end{eqnarray}
In fact, if one of  $c_1$, $c_2$, $c_3$, $c_4$, $c_5$ is $0$ or $1$ then $\cT_{AB:CD}$ is not a UPB from Lemma \ref{le:size8>=3} (i) or $U_{11}=U_{12}$. If two of  $c_1$, $c_2$, $c_3$, $c_4$ are $1$'s, we can assume $c_1=c_2=1$. Then the space spanned by $\ket{h_4,i_4}$, $\ket{h_5,i_5}$, $\ket{h_6,i_6}$, $\ket{h_7,i_7}$ has dimension at most three. Also the space spanned by $\ket{f_1,g_1}$, $\ket{f_2,g_2}$, $\ket{f_3,g_3}$, $\ket{f_8,g_8}$ has dimension at most three. So there exist $\phi$, $\psi\in\bbC^4$ such that $\phi$ and $\psi$ is respectively orthogonal to $\ket{f_1,g_1}$, $\ket{f_2,g_2}$, $\ket{f_3,g_3}$, $\ket{f_8,g_8}$ and $\ket{h_4,i_4}$, $\ket{h_5,i_5}$, $\ket{h_6,i_6}$, $\ket{h_7,i_7}$. It makes a contradiction with the assumption that $\cT_{AB:CD}$ is a UPB. Therefore, the form $U_{11}$ and $U_{12}$ are all possible cases. We have proved the claim above \eqref{eq:startclass3}. 

For $U_{11}$, since the first two rows are orthogonal to row $3,5,6,7,8$, we obtain that $\ket{0,0}\in\cH_{AB}$ is orthogonal to $\ket{f_j,g_j}$ for  $j=3,5,6,7,8$. Namely, three of $f_3$, $f_5$, $f_6$, $f_7$, $f_8$ or $g_3$, $g_5$, $g_6$, $g_7$, $g_8$ are $1$'s. It's a contradiction with Lemma \ref{le:size8>=3} (iii) and the fact that $\cT_{AB:CD}$ is a UPB of size $8$.

For $U_{12}$, we have $f_3, g_3, f_4, g_4, f_5, g_5, f_6, g_6, f_7, g_7, f_8, g_8\neq 0$ from Lemma \ref{le:3f3h}(ii). And the first rows are orthogonal to each of the last six rows of $U_{12}$. So $\ket{0,0}\in\cH_{AB}$ is orthogonal to each of $\ket{f_3,g_3}$, $\ket{f_6,g_6}$, $\ket{f_7,g_7}$, $\ket{f_8,g_8}$. Also the first two rows both have at most two identical product vectors from Lemma \ref{le:size8>=3} (iii). So we have the two cases $U_{121}$ and $U_{122}$.
\begin{eqnarray}
U_{121}=
\bma
0 & 0 & 0 & 0\\
0 & 0 & 0 & 1\\
1 & g_3 & 0 & i_3\\
f_4 & g_4 & 1 & i_4\\
f_5 & g_5 & 1 & i_5\\
1 & g_6 & a_5 & i_6\\
f_7 & 1 & a_5 & i_7\\
f_8 & 1 & a_5' & i_8\\
\ema, 
\text{\quad}
U_{122}=
\bma
0 & 0 & 0 & 0\\
0 & 0 & 0 & 1\\
1 & g_3 & 0 & i_3\\
f_4 & g_4 & 1 & i_4\\
f_5 & g_5 & 1 & i_5\\
f_6 & 1 & a_5 & i_6\\
f_7 & 1 & a_5 & i_7\\
1 & g_8 & a_5' & i_8\\
\ema.
\end{eqnarray}
For $U_{121}$, row 3 is orthogonal to row 7 and 8, so $\ket{i_3}$ is orthogonal to $\ket{i_7}$ and $\ket{i_8}$. That is, $i_7=i_8=i_3'$. Using the similar argument, we can obtain $i_6=i_7=i_3'$ in $U_{122}$ because row 3 is orthogonal to row 6 and 7. So $U_{121}$ and $U_{122}$ are respectively equivalent to $U_{1211}$ and $U_{1221}$.
\begin{eqnarray}
U_{1211}=
\bma
0 & 0 & 0 & 0\\
0 & 0 & 0 & 1\\
1 & g_3 & 0 & i_3\\
f_4 & g_4 & 1 & i_4\\
f_5 & g_5 & 1 & i_5\\
1 & g_6 & a_5 & i_6\\
f_7 & 1 & a_5 & i_3'\\
f_8 & 1 & a_5' & i_3'\\
\ema, 
\text{\quad}
U_{1221}=
\bma
0 & 0 & 0 & 0\\
0 & 0 & 0 & 1\\
1 & g_3 & 0 & i_3\\
f_4 & g_4 & 1 & i_4\\
f_5 & g_5 & 1 & i_5\\
f_6 & 1 & a_5 & i_3'\\
f_7 & 1 & a_5 & i_3'\\
1 & g_8 & a_5' & i_8\\
\ema.
\end{eqnarray}
 So column $3$ and $4$ of $U_{1211}$ and $U_{1221}$ shows a contradiction with Lemma \ref{le:size8>=3} (iii) and the fact that $\cT_{AB:CD}$ is a UPB of size $8$. 
\end{proof}

The following proposition can be proven similarly to Lemma \ref{le:size8>=3}, \ref{le:independent}, \ref{le:44=a3b3}, \ref{le:3f3h}.
\begin{proposition}
\label{pp:0of2*3}
Let $\cT_{A:B:C:D}$=\{$\ket{f_1,g_1,h_1,i_1}$, $\ket{f_2,g_2,h_2,i_2}$,...,$\ket{f_8,g_8,h_8,i_8}$\} be a 4-qubit UPB of size $8$ in $\cH_A\ox\cH_B\ox\cH_C\ox\cH_D$. 

(i) If there are two subscripts $j_1,j_2$ such that $\ket{f_{j_1}}=\ket{f_{j_2}}$, $\ket{g_{j_1}}=\ket{g_{j_2}}$ and $\ket{h_{j_1}}=\ket{h_{j_2}}$, then $\cT_{A:B:C:D}$ is no longer a UPB in $\cH_{AB}\ox\cH_{CD}$.

(ii)  If there are two subscripts $j_1,j_2$ such that $\ket{f_{j_1}}=\ket{f_{j_2}}$, $\ket{g_{j_1}}=\ket{g_{j_2}}$, $\ket{h_{j_1}}=\ket{h_{j_2}'}$ and $\ket{i_{j_1}}=\ket{i_{j_2}'}$, then $\cT_{A:B:C:D}$ is no longer a UPB in $\cH_{AB}\ox\cH_{CD}$.

(iii) If there are two subscripts $j_1,j_2$ such that $\ket{f_{j_1}}=\ket{f_{j_2}}$, $\ket{g_{j_1}}=\ket{g_{j_2}'}$, $\ket{h_{j_1}}=\ket{h_{j_2}}$ and $\ket{i_{j_1}}=\ket{i_{j_2}'}$, then $F_1$ in Appendix \ref{app:uom,f1-f6} happen to be all UOMs satisfying this condition.

(iv) If there are two subscripts $j_1,j_2$ such that $\ket{f_{j_1}}=\ket{f_{j_2}}$ and $\ket{g_{j_1}}=\ket{g_{j_2}}$, then $\cT_{A:B:C:D}$ is no longer UPB in $\cH_{AB}\ox\cH_{CD}$. 

(v)  If there are two subscripts $j_1,j_2$ such that $\ket{f_{j_1}}=\ket{f_{j_2}}$ and $\ket{h_{j_1}}=\ket{h_{j_2}}$,  then $F_2,F_3,F_4,F_5$ in Appendix \ref{app:uom,f1-f6} happen to be all UOMs satisfying this condition.

(vi)  If there are five distinct subscripts $j_1,j_2,j_3,j_4,j_5$ such that $\ket{f_{j_1}}=\ket{f_{j_2}}=\ket{f_{j_3}}$ and $\ket{h_{j_3}}=\ket{h_{j_4}}=\ket{h_{j_5}}$, then $F_2,F_3,F_4,F_5$ in Appendix \ref{app:uom,f1-f6} happen to be all UOMs satisfying this condition.

(vii)  If there are three distinct subscripts $j_1,j_2,j_3$ such that $\ket{f_{j_1}}=\ket{f_{j_2}}=\ket{f_{j_3}}$, then $F_2,F_3,F_4,F_5$  in Appendix \ref{app:uom,f1-f6} happen to be all UOMs satisfying this condition.

(viii)  If there are three distinct subscripts $j_1,j_2$ such that $\ket{f_{j_1}}=\ket{f_{j_2}}$, then $F_2,F_3,F_4,F_5,F_6$ in Appendix \ref{app:uom,f1-f6} happen to be all UOMs satisfying this condition.

\end{proposition}

\bibliographystyle{unsrt}

\bibliography{kaipaper2}

\end{document}